\documentclass[11pt]{article}

\usepackage{amsmath}
\usepackage{graphicx}
\usepackage{amsthm}
\usepackage{hyperref}
\usepackage{amssymb}
\usepackage{fullpage}

\newcommand{\ket}[1]{|#1\rangle}
\newcommand{\bra}[1]{\langle #1|}
\newcommand{\proj}[1]{\ket{#1} \bra{#1}}

\newcommand{\vectornorm}[1]{\left|\left|#1\right|\right|}
\newcommand{\eps}{\epsilon}
\renewcommand{\Pr}{\mathbb{P}}
\newcommand{\Expect}{\mathbb{E}}
\newcommand{\tr}{{\rm tr}\,}
\newcommand{\swap}{\mathcal{F}}
\newcommand{\ot}{\otimes}
\newcommand{\cE}{\mathcal{E}}
\newcommand{\cG}{\mathcal{G}}
\newcommand{\be}{\begin{equation}}
\newcommand{\ee}{\end{equation}}
\newcommand{\stocdomr}{\trianglelefteq}
\newcommand{\stocdoml}{\trianglerighteq}
\newcommand{\majr}{\preceq}
\newcommand{\majl}{\succeq}
\def\ba#1\ea{\begin{align}#1\end{align}}
\def\bit{\begin{itemize}}
\def\eit{\end{itemize}}
\def\l{\left}
\def\r{\right}
\def\<{\langle}
\def\>{\rangle}

\newtheorem{theorem}{Theorem}
\newtheorem{lemma}[theorem]{Lemma}
\newtheorem{definition}[theorem]{Definition}
\newtheorem{corollary}[theorem]{Corollary}
\newtheorem{conjecture}[theorem]{Conjecture}

\numberwithin{equation}{section}
\numberwithin{theorem}{section}

\newcommand{\nn}{\nonumber\\}
\newcommand{\eq}[1]{Eqn.~\ref{eq:#1}}
\newcommand{\peq}[1]{(Eqn.~\ref{eq:#1})}
\newcommand{\thmref}[1]{Theorem \ref{thm:#1}}
\newcommand{\corref}[1]{Corollary \ref{cor:#1}}
\newcommand{\lemref}[1]{Lemma \ref{lem:#1}}
\newcommand{\secref}[1]{Section \ref{sec:#1}}
\newcommand{\defref}[1]{Definition \ref{def:#1}}

\DeclareMathOperator{\poly}{poly}
\def\ra{\rightarrow}
\def\bbC{\mathbb{C}}
\def\bbE{\mathbb{E}}
\newcommand{\bfq}{\mathbf{q}}
\newcommand{\bfp}{\mathbf{p}}

\begin{document}

\title{Random Quantum Circuits are Approximate 2-designs}

\author{Aram W.~Harrow and Richard A.~Low\footnote{low@cs.bris.ac.uk}\\ \emph{\small Department of Computer Science, University of Bristol, Bristol, U.K.}}
 
\maketitle

\begin{abstract}
Given a universal gate set on two qubits, it is well known that
applying random gates from the set to random pairs of qubits will
eventually yield an approximately Haar-distributed unitary.  However,
this requires exponential time.  We show that random circuits of
only polynomial length will approximate the first and second moments
of the Haar distribution, thus forming approximate 1- and 2-designs.
Previous constructions required longer circuits and worked only for
specific gate sets.  As a corollary of our main result, we also
improve previous bounds on the convergence rate of random walks on
the Clifford group. 
\end{abstract}

\section{Introduction: Pseudo-Random Quantum Circuits}

There are many examples of algorithms that make use of random states
or unitary operators (e.g.~\cite{AmbainisSmith04,Sen05}).  However,
exactly sampling from the uniform Haar distribution is inefficient.
In many cases, though, only pseudo-random operators are required.  To
quantify the extent to which the pseudo-random operators behave like
the uniform distribution we use the notion of \emph{$k$-designs}
(often referred to as $t$-designs).  A $k$-design has $k^{\text{th}}$
moments equal to those of the Haar distribution.  For most uses of
random states or unitaries, this is sufficient.  Constructions of
exact $k$-designs on states are known (see \cite{AmbainisEmerson07}
and references therein) and some are efficient.  Ambainis and Emerson
\cite{AmbainisEmerson07} introduced the notion of approximate state
$k$-designs, which can be implemented efficiently for any $k$.
However, the known constructions of unitary $k$-designs are
inefficient to implement.  Approximate unitary $2$-designs have been
considered \cite{DLT02,DCEL06,GAE07}, although the approaches are
specific to 2-designs. 

We consider a general class of random circuits where a series of
two-qubit gates are chosen from a universal gate set.  We give a
framework for analysing the $k^{\text{th}}$ moments of these circuits.
Our conjecture, based on an analogous classical result
\cite{HooryBrodsky04}, is that a random circuit on $n$ qubits of
length $\poly(n,k)$ is an approximate $k$-design.  While we do not
prove this, we instead give a tight analysis of the $k=2$ case.  We
find that in a broad class of natural random circuit models (described
in \secref{RandomCircuits}), a circuit of length $O(n(n+\log 1/\eps))$
yields an $\eps$-approximate 2-design.  Our definition of an
approximate $k$-design is in \secref{k-design-def}.  Our results also apply to an alternative definition of an approximate 2-design from \cite{DCEL06}, for which we show random circuits of length $O(n(n+\log 1/\eps))$ yield $\eps$-approximations, thus extending the results of that paper to a larger
class of circuits.  Moreover, our results also apply to random
stabiliser circuits, meaning that a random stabiliser circuit of
length $O(n(n+\log 1/\eps))$ will be an
$\eps$-approximate 2-design.  This both simplifies the construction
and tightens the efficiency of the approach of \cite{DLT02}, which
constructed $\eps$-approximate 2-designs in time $O(n^6(n^2+\log
1/\eps))$ using $O(n^3)$ elementary quantum gates.

\subsection{Random Circuits}
\label{sec:RandomCircuits}

The random circuit we will use is the following.  Choose a 2-qubit
gate set that is universal on $U(4)$ (or on the stabiliser subgroup of
$U(4)$).  One example of this is the set of all one qubit gates
together with the controlled-NOT gate.  Another is simply the set of
all of $U(4)$.  Then, at each step, choose a random pair of qubits and
apply a gate from the universal set chosen uniformly at random.  For
the $U(4)$ case, the distribution will be the Haar measure on $U(4)$.
One such circuit is shown in Fig.~\ref{figRandomCircuit} for $n=4$
qubits.  This is based on the approach used in Refs.~\cite{ODP06,DOP07} but our analysis is both simpler and more general.

\begin{figure}[h]
  \begin{center}
    \includegraphics[width=12cm]{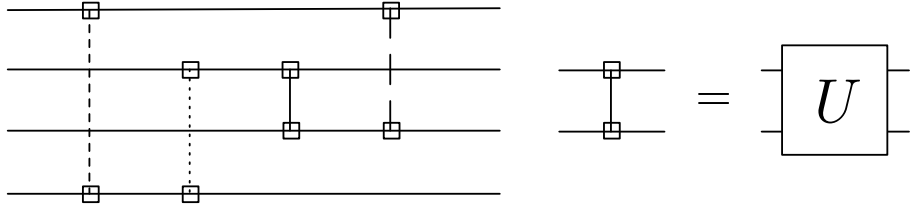}
    \caption{An example of a random circuit.  Different lines indicate
      a different gate is applied at each step.} 
    \label{figRandomCircuit}
  \end{center}
\end{figure}
Since the universal set can generate the whole of $U(2^n)$ in this
way, such random circuits can produce any unitary.  Further, since
this process converges to a unitarily invariant distribution and the
Haar distribution is unique, the resulting unitary must be uniformly
distributed amongst all unitaries \cite{ELL05}.  Therefore this
process will eventually converge to a Haar distributed unitary from
$U(2^n)$.  This is proven rigourously in \lemref{GEigenvalues}.  However,
a generic element of $U(2^n)$ has $4^n$ real parameters, and thus to
even have $\Omega(4^{-n})$ fidelity with the Haar distribution
requires $\Omega(4^n)$ 2-qubit unitaries.  We address this 
problem by considering only the lower-order moments of the
distribution and showing these are nearly the same for random circuits
as for Haar-distributed unitaries.  This claim is formally described
in \thmref{Main2Design}.

Our paper is organised as follows.  In \secref{Preliminaries} we
define unitary $k$-designs and explain how a random circuit could be
used to construct a $k$-design.  In \secref{Moments} we work out how the state evolves after a single step of the random circuit.  We then extend this to multiple steps in \secref{Convergence} and prove our general convergence results.  A key simplification
will be (following \cite{ODP06}) to map the evolution of the second
moments of the quantum circuit onto a classical Markov chain.  We then
prove a tight convergence result for the case where the gates are
chosen from $U(4)$ in \secref{U4Convergence}.  This section contains
most of the technical content of the paper.  Using our bounds on
mixing time we put together the proof that random circuits yield approximate unitary
2-designs in \secref{MainResult}.  \secref{Conclusion} concludes with
some discussion of applications.

\section{Preliminaries}
\label{sec:Preliminaries}
\subsection{Pauli expansion}
Much of the following will be done in the Pauli basis.  The Pauli
operators will be taken as $\{ \sigma_0, \sigma_1, \sigma_2, \sigma_3
\}$ and defined to be
$$\sigma_0 = \begin{pmatrix}1 & 0 \\ 0 & 1\end{pmatrix} 
\qquad
\sigma_1 = \begin{pmatrix} 0& 1 \\ 1 & 0\end{pmatrix} 
\qquad
\sigma_2 = \begin{pmatrix} 0& -i \\ i & 0\end{pmatrix} 
\qquad
\sigma_3 = \begin{pmatrix} 1& 0 \\ 0 & -1\end{pmatrix} 
$$
If $\ket{\psi}\in\mathbb{C}^{2^n}$ is a state on $n$ qubits then we write
$\psi=\proj{\psi}$.  We can expand $\psi$ in the Pauli basis as
\begin{equation}
\psi = 2^{-n/2} \sum_p \gamma(p) \sigma_p
\end{equation}
where $\sigma_p = \sigma_{p_1} \otimes \ldots \otimes \sigma_{p_n}$ for the string $p = p_1 \ldots p_n$.  Inverting, the coefficients $\gamma(p)$ are given by
\begin{equation}
\gamma(p) = 2^{-n/2} \tr \sigma_p \psi.
\end{equation}
It is easy to show that the coefficients $\gamma(p)$ are real and, with the chosen normalisation, the squares sum to $\tr \psi^2$, which is 1 for pure
$\psi$.  In general
\begin{equation*}
\sum_p \gamma^2(p) \le 1
\end{equation*}
with equality if and only if $\psi$ is pure.  Note also that $\tr \psi = 1$ is equivalent to $\gamma(0) = 2^{-n/2}$.

This notation is extended to states on $nk$ qubits by treating $\gamma$ as a
function of $k$ strings from $\{0,1,2,3\}^n$.  Thus a 
 state $\rho$ on $nk$ qubits is written as
\begin{equation}
\label{eq:PauliBasisGeneral}
\rho = 
2^{-nk/2} \sum_{p_1, \ldots, p_k} \gamma_0(p_1, \ldots, p_k) \sigma_{p_1} \otimes \ldots \otimes \sigma_{p_k}.
\end{equation}

\subsection{\texorpdfstring{$k$}{k}-designs}
\label{sec:k-design-def}

We will say that a $k$-design is efficient if the effort required to
sample a state or unitary from the design is polynomial in $n$ and
$k$.  Note that we do not require the number of states to be
polynomial because, even for approximate unitary designs, an
exponential number of unitaries is required.  Rather, the number of
random bits needed to specify an element of the design should be $\poly(n,k)$.

\subsubsection{State designs}

A (state) $k$-design is an ensemble of states such that, when one state is
chosen from the ensemble and copied $k$ times, it is indistinguishable
from a uniformly random state.  This is a way of quantifying the
pseudo-randomness of the state and is a quantum analogue of $k$-wise
independence.  Hayashi et al.~\cite{HHM06} give an inefficient construction of $k$-designs for any $n$ and $k$.

The state $k$-design definition we use is due to Ref.~\cite{AmbainisEmerson07}:
\begin{definition}
An ensemble of quantum states $\{ p_i, \ket{\psi_i} \}$ is a state $k$-design if
\begin{equation}
\sum_i p_i \left( \ket{\psi_i} \bra{\psi_i} \right)^{\otimes k} = \int_\psi \left( \ket{\psi} \bra{\psi} \right)^{\otimes k} d\psi
\end{equation}
where the integration is taken over the left invariant Haar
measure on the unit sphere in $\mathbb{C}^d$, normalised so that $\int_\psi d\psi = 1$. 
\end{definition}
It is well known that the above integral is equal to
$\frac{\Pi_{+k}}{{k+d-1 \choose k}}$, where $\Pi_{+k}$ is the
projector onto the symmetric subspace of $k$ $d$-dimensional spaces.
For a rigourous proof, see Ref.~\cite{GoodmanWallach98} and for a less
precise proof but from a quantum information perspective see
Ref.~\cite{BBDEJM97}.

\subsubsection{Unitary designs}

A unitary $k$-design is, in a sense, a stronger version of a state
design.  Just as applying a Haar-random unitary to an arbitrary pure
state results in a uniformly random pure state, applying a unitary
chosen from a unitary $k$-design to an arbitrary pure state should
result in a state $k$-design.  Another way to say this is that the state obtained by acting $U^{\ot k}$,
where $U$ is drawn from a unitary $k$-design on
$U(d)$, on any $d^k$-dimensional state should be indistinguishable from the case where $U$ is drawn uniformly from $U(d)$.  Formally,
we have:

\begin{definition}
\label{def:UnitaryDesign}
Let $\{ p_i, U_i \}$ be an ensemble of unitary operators.  Define
\begin{equation}
\label{eq:UnitaryDesign}
\cG_{W}(\rho) = \sum_i p_i U_i^{\otimes k} \rho (U_i^\dagger)^{\otimes k}
\end{equation}
and
\begin{equation}
\cG_H(\rho) = \int_U U^{\otimes k} \rho (U^\dagger)^{\otimes k} dU.
\end{equation}
Then the ensemble is a unitary $k$-design iff $\cG_{W} = \cG_H$.
\end{definition}
 Unitary designs can also be defined in terms of polynomials, so that
if $p$ is a polynomial with degree $k$ in the matrix elements of $U$
and $k$ in the matrix elements of $U^*$, then averaging $p$ over a
unitary $k$-design should give the same answer as averaging over the
Haar measure.  To see the equivalence with \defref{UnitaryDesign} note
that averaging a monomial over our ensemble can be expressed as
$\bra{i_1,\ldots,i_k}
\cG_{W}(\ket{j_1,\ldots,j_k}\bra{j_1',\ldots,j_k'})
\ket{i_1',\ldots,i_k'}$, and so if $\cG_{W}=\cG_H$ then
any polynomial of degree $k$ will have the same expectation over both
distributions.

\subsection{Approximate \texorpdfstring{$k$}{k}-designs}

\subsubsection{Approximate state designs}

Numerous examples of exact efficient state 2-design constructions are known (e.g.~\cite{Barnum02}) but general exact constructions are not efficient in $n$ and $k$.  Approximate state designs were first introduced by Ambainis and Emerson \cite{AmbainisEmerson07} and they constructed efficient approximate state $k$-designs for any $k$.  Aaronson \cite{Aaronson07} also gives an efficient approximate construction.

We define approximate state designs as follows. 
\begin{definition}
\label{def:ApproxStateDesign}
An ensemble of quantum states $\{ p_i, \ket{\psi_i} \}$ is an $\eps$-approximate state $k$-design if
\begin{equation}
\label{eq:ApproxStateDesign}
(1-\eps)\int_\psi \left( \ket{\psi} \bra{\psi} \right)^{\otimes k} d\psi \le \sum_i p_i \left( \ket{\psi_i} \bra{\psi_i} \right)^{\otimes k} \le (1+\eps) \int_\psi \left( \ket{\psi} \bra{\psi} \right)^{\otimes k} d\psi
\end{equation}
\end{definition}
In \cite{AmbainisEmerson07}, a similar definition was proposed but
with the additional requirement that the ensemble also forms a
1-design (exactly), i.e.
\begin{equation*}
\sum_i p_i \ket{\psi_i} \bra{\psi_i} = \int_\psi \ket{\psi} \bra{\psi} d\psi
\end{equation*}
This requirement was necessary there only so that a suitably normalised
version of the ensemble would form a POVM. We will not use it.

 By taking the partial trace one can show that a
$k$-design is a $k'$-design for $k' \le k$.  Thus approximate
$k$-designs are always at least approximate 1-designs.

\subsubsection{Approximate unitary designs}

It was shown in Ref.~\cite{AMTW00} that a
quantum analogue of a one time pad requires $2n$ bits to exactly
randomise an $n$ qubit state.  However, in Ref.~\cite{AmbainisSmith04}
it was shown that $n+o(n)$ bits suffice to do this approximately.
Translated into $k$-design language, this says an exact unitary
1-design requires $2^{2n}$ unitaries but can be done approximately
with $2^{n+o(n)}$.  So approximate designs can have fewer unitaries than exact designs.  Here, we are interested in improving the efficiency of implementing the unitaries.  There are no known efficient exact constructions of unitary $k$-designs; it is hoped that our approach will yield approximate unitary designs efficiently.

We will require approximate unitary $k$-designs to be close in the
diamond norm \cite{KSV02}: 
\begin{definition}
\label{def:diamondNorm}
The diamond norm of a superoperator $T$ 
\begin{equation*}
\vectornorm{T}_{\diamond} = \sup_d \vectornorm{T \otimes \text{id}_d}_{\infty} = 
\sup_d \sup_{X \ne 0} \frac{\vectornorm{(T \otimes
    \text{id}_d)X}_{1}}{\vectornorm{X}_{1}} 
\end{equation*}
where $\text{id}_d$ is the identity channel on $d$ dimensions.
\end{definition}
Operationally, the diamond norm of the difference between two quantum
operations tells us the largest possible probability of distinguishing
the two operations if we are allowed to have them act on part of an
arbitrary, possibly entangled, state.  In the supremum over ancilla
dimension $d$, it can be shown that $d$ never needs to be larger than
the dimension of the system that $T$ acts upon.   The diamond norm is
closely related to completely bounded norms (cb-norms), in that
$\vectornorm{T}_\diamond$ is the cb-norm of $T^\dag$ and can also be
interpreted as the $L_1\ra L_1$ cb-norm of $T$ itself \cite{DJKR,Paulsen}. 

We can now define approximate unitary $k$-designs.
\begin{definition}
\label{def:ApproxUnitaryDesign}
$\cG_{W}$ is an $\eps$-approximate unitary $k$-design if
\begin{equation}
\label{eq:ApproxUnitaryDesign}
\vectornorm{\cG_{W} - \cG_H}_{\diamond} \le \eps,
\end{equation} 
where $\cG_{W}$ and $\cG_H$ are defined in \defref{UnitaryDesign}.
\end{definition}
In Ref.~\cite{DCEL06}, they consider approximate twirling, which is implemented using an approximate 2-design.  They give an alternative definition of closeness which is more convenient for this application:
\begin{definition}[\cite{DCEL06}]
\label{def:DankertApprox}
Let $\{ p_i, U_i \}$ be an ensemble of unitary operators.  Then this ensemble is an $\eps$-approximate twirl if
\begin{equation}
\max_\Lambda \vectornorm{\Expect_W W(\Lambda(W^\dagger \rho W))W^\dagger - \Expect_U U(\Lambda(U^\dagger \rho U))U^\dagger}_{\diamond} \le \frac{\eps}{d^2}
\end{equation}
where the first expectation is over $W$ chosen from the ensemble and the second is the Haar average.  The maximisation is over channels $\Lambda$ and $d$ is the dimension ($2^n$ in our case).
\end{definition}
Our results work for both definitions with the same efficiency.

\subsection{Random Circuits as \texorpdfstring{$k$}{k}-designs}

If a random circuit is to be an approximate $k$-design then
\eq{ApproxUnitaryDesign} must be satisfied where the $U_i$
are the different possible random circuits.  We can think of this as
applying the random circuit not once but $k$ times to $k$ different
systems.  

Suppose that applying $t$ random gates yields the random
circuit $W$.  If $W^{\ot k}$ acts on an $nk$-qubit state $\rho$, then
following the notation of \eq{ApproxUnitaryDesign}, the resulting
state is 
\begin{equation}
\label{eqn:StateEvolutionPauliBasis}
\rho_W := W^{\ot k}\rho(W^\dag)^{\ot k} =
2^{-nk/2} \sum_{p_1, \ldots, p_k} \gamma_0(p_1, \ldots, p_k) W
\sigma_{p_1} W^{\dagger} \otimes \ldots \otimes W \sigma_{p_k}
W^{\dagger}. 
\end{equation}

For this to be a $k$-design, the expectation over
all choices of random circuit should match the expectation over
Haar-distributed $W\in U(2^n)$.

We are now ready to state our main results.  Our results apply to a
large class of gate sets which we define below: 
\begin{definition}
\label{def:2copy-gapped}
Let $\cE=\{p_i,U_i\}$ be a discrete ensemble of elements from $U(d)$.
Define an operator $G_\cE$ by
\be G_\cE := \sum_i p_i U_i^{\ot k} \ot (U_i^*)^{\ot k}
\label{eq:gapped-condition}\ee
More generally, we can consider continuous distributions.  If $\mu$ is
a probability measure on $U(d)$ then we can define $G_\mu$ by analogy
as
\be G_\mu := \int_{U(d)} d\mu(U) U^{\ot k} \ot (U^*)^{\ot k}
\label{eq:gapped-condition2}\ee
Then
$\cE$ (or $\mu$) is $k$-copy gapped if $G_\cE$ (or $G_\mu$)
has only $k!$ eigenvalues with absolute value equal to $1$.
\end{definition}
For any discrete ensemble $\cE=\{p_i,U_i\}$, we can define a measure
$\mu=\sum_i p_i \delta_{U_i}$.  Thus, it suffices to state our theorems in
terms of $\mu$ and $G_\mu$.  

The condition on $G_\mu$ in the above definition may seem somewhat
strange.  We will see in \secref{Moments} that when $d\geq k$ there is
a $k!$-dimensional subspace of $(\bbC^d)^{\ot 2k}$ that is acted upon
trivially by any $G_\mu$.  Additionally, when $\mu$ is the Haar
measure on $U(d)$ then $G_\mu$ is the projector onto this space.
Thus, the $k$-copy gapped condition implies that vectors orthogonal to
this space are shrunk by $G_\mu$.

We will see that $G_\mu$ is $k$-copy gapped in a  number of important cases.
First, we give a definition of universality that can apply not only to
discrete gates sets, but to arbitrary measures on $U(4)$.

\begin{definition}\label{def:dist-universal}
Let $\mu$ be a distribution on $U(4)$.  Suppose that for any open ball
$S\subset U(4)$ there exists a positive integer $\ell$
such that $\mu^{\star \ell}(S)>0$.  Then we say $\mu$ is universal
[for $U(4)$].
\end{definition}
  Here $\mu^{\star \ell}$ is the $\ell$-fold convolution
of $\mu$ with itself; i.e.
$$\mu^{\star \ell}= \int \delta_{U_1\cdots U_\ell} 
d\mu(U_1)\cdots d\mu(U_\ell).$$
When $\mu$ is a discrete distribution over a set $\{U_i\}$,
\defref{dist-universal} is equivalent to the usual definition of
universality for a finite set of unitary gates.

\begin{theorem}\label{thm:kCopyGappedExamples}
The following distributions on $U(4)$ are $k$-copy gapped:
\begin{itemize}
\item[(i)]{Any universal gate set.  Examples are $U(4)$ itself, any
    entangling gate together with all single qubit gates, or the gate
    set considered in \cite{ODP06}.}
\item[(ii)]{Any approximate (or exact) unitary $k$-design on 2 qubits, such as the
    uniform distribution over the 2-qubit Clifford group, which is an exact 2-design.} 
\end{itemize}
\end{theorem}
\begin{proof}
\mbox{}
\begin{itemize}
\item[(i)]{This is proven in \lemref{GEigenvalues}.}
\item[(ii)]{This follows straight from \defref{UnitaryDesign}.}\qedhere
\end{itemize}
\end{proof}
\begin{theorem}
\label{thm:Main2Design}
%Let $\cE=\{p_i,U_i\}$ be a 2-copy gapped ensemble.  
Let $\mu$ be a 2-copy gapped distribution and $W$ be a random circuit on $n$ qubits
obtained by drawing $t$ random unitaries according to $\mu$ and
applying each of them to a random pair of qubits.  Then there exists $C$ (depending only on $\mu$) such that for any $\eps>0$ and any $t \geq C(n(n+\log 1/\eps))$, $\cG_W$ is an
  $\eps$-approximate unitary 2-design according to either 
  \defref{ApproxUnitaryDesign} or \defref{DankertApprox}.
\end{theorem}

To prove \thmref{Main2Design}, we show that the second
moments of the random circuits converge quickly to those of a uniform
Haar distributed unitary.  For $W$ a circuit as in \thmref{Main2Design}, write $\gamma_W(p_1, p_2)$ for the Pauli coefficients of $\rho_W = W^{\ot 2} \rho \l(W^\dagger\r)^{\ot 2}$.  Then write $\gamma_t(p_1, p_2) = \Expect_W \gamma_W(p_1, p_2)$ where $W$ is a circuit of length $t$.  Then we have
\begin{lemma}
\label{lem:MainMixing}
Let $\mu$ and $W$ be as in \thmref{Main2Design}.  Let the initial state be $\rho$ with $\gamma_0(p,p) \ge 0$ and $\sum_p \gamma_0(p,p) = 1$ (for example the state $\proj{\psi} \otimes \proj{\psi}$ for any pure state $\ket{\psi}$).  Then there
exists a constant $C$ (possibly depending on $\mu$) such that for
any $\eps>0$
\begin{itemize}
\item[(i)]{\begin{equation}
\sum_{p_1, p_2 \atop p_1 p_2 \ne 00} \left( \gamma_t(p_1,
    p_2) - \delta_{p_1 p_2} \frac{1}{2^n(2^n+1)} \right)^2 \le
\eps 
\end{equation}
for $t \ge Cn \log 1/\eps$.}
\item[(ii)]{\begin{equation}
\sum_{p_1, p_2 \atop p_1 p_2 \ne 00} \left| \gamma_t(p_1, p_2)
  - \delta_{p_1 p_2} \frac{1}{2^n(2^n+1)} \right| \le \eps
\end{equation}
for $t \ge Cn(n + \log 1/\eps)$ or, when $\mu$ is the uniform
distribution on $U(4)$ or its stabiliser subgroup, $t \ge Cn \log
\frac{n}{\eps}$.} 
\end{itemize}
\end{lemma}
We can then extend this to all states by a simple corollary:
\begin{corollary}
\label{cor:MainMixing}
Let $\mu$, $W$ and $\gamma_W$ be as in \lemref{MainMixing}.  Then, for any initial state $\rho = \frac{1}{2^n}\sum_{p_1,p_2} \gamma_0 (p_1,p_2) \sigma_{p_1} \ot \sigma_{p_2}$, there
exists a constant $C$ (possibly depending on $\mu$) such that for
any $\eps>0$
\begin{itemize}
\item[(i)]{\begin{equation}
\sum_{p_1, p_2 \atop p_1 p_2 \ne 00} \left( \gamma_t(p_1,
    p_2) - \delta_{p_1 p_2} \frac{\sum_{p \ne 0} \gamma_0 (p,p)}{4^n-1} \right)^2 \le
\eps 
\label{eq:2-norm-converge} 
\end{equation}
for $t \ge Cn(n+\log 1/\eps)$.}
\item[(ii)]{\begin{equation}
\sum_{p_1, p_2 \atop p_1 p_2 \ne 00} \left| \gamma_t(p_1, p_2)
  - \delta_{p_1 p_2} \frac{\sum_{p \ne 0} \gamma_0 (p,p)}{4^n-1} \right| \le \eps
\label{eq:1-norm-converge}
\end{equation}
for $t \ge Cn(n + \log 1/\eps)$.} 
\end{itemize}
\end{corollary}

By the usual definition of an approximate
design (\defref{ApproxUnitaryDesign}), we only need 
convergence in the 2-norm \peq{2-norm-converge}, which is implied by
1-norm convergence \peq{1-norm-converge} but 
weaker.  However, \defref{DankertApprox}, which requires the map to be close to the
twirling operation, requires 1-norm convergence
(i.e. \eq{1-norm-converge}).  Thus, \thmref{Main2Design} for \defref{ApproxUnitaryDesign}
follows from \corref{MainMixing}(i) 
and \thmref{Main2Design} for \defref{DankertApprox} follows from \corref{MainMixing}(ii).
\thmref{Main2Design} is proved in \secref{MainResult} and \corref{MainMixing} in \secref{Convergence}.

We note that, in the course of proving \lemref{MainMixing}, we prove that the eigenvalue gap (defined in \secref{MarkovChain}) of the Markov chain that gives the evolution of the $\gamma(p,p)$ terms is $O(1/n)$.  It is easy to show that this bound is tight for some gate sets.

{\em Related work:}  Here we summarise the other efficient
constructions of approximate unitary 2-designs.
\begin{itemize}
\item The uniform distribution over the Clifford group on $n$ qubits
  is an exact 2-design \cite{DLT02}.  Moreover, \cite{DLT02} described
  how to sample from the Clifford group using $O(n^8)$ classical gates
  and $O(n^3)$ quantum gates.  Our results show that applying
  $O(n(n+\log 1/\eps))$ random two-qubit Clifford gates also achieve
  an $\eps$-approximate 2-design (although not necessarily a
  distribution that is within $\eps$ of uniform on the Clifford
  group).
\item Dankert et al.~\cite{DCEL06} gave a specific circuit
  construction of an approximate 2-design.  To achieve small error in
  the sense of \defref{ApproxUnitaryDesign}, their circuits require
  the same
  $O(n(n+\log 1/\eps))$ gates that our random circuits do.  However,
  when we use 
  \defref{DankertApprox}, the circuits from \cite{DCEL06} only
  need $O(n \log 1/\eps)$ gates while the random circuits analysed
  in this paper need to be
  length $O(n(n+\log 1/\eps))$.
\item The closest results to our own are in the papers by Oliveira et
  al.~\cite{ODP06,DOP07}, which considered a specific gate set
  (random single qubit gates and a controlled-NOT) and proved that the
  second moments converge in time $O(n^2(n+ \log 1/\eps))$.  Our
  strategy of analysing random
  quantum circuits in terms of classical Markov chains is also adapted
  from \cite{ODP06,DOP07}.  In \secref{Moments}, we generalise this approach
   to analyse the $k^{\text{th}}$ moments for arbitrary $k$.

  The main results of our paper extend the results of  \cite{ODP06,DOP07} to a larger class of gate sets and improve their
    convergence bounds.  Some of these improvements have been conjectured by
    \cite{Znidaric07}, which presented numerical evidence in support
    of them.
\end{itemize}

\section{Analysis of the Moments}
\label{sec:Moments}

In order to prove our results, we need to understand how the state
evolves after each step of the random circuit.  
%We need to average over the choices of pairs and gates at each step to
%find out how the expected coefficients evolve.
In this section we consider just one step and a fixed pair of qubits.
Later on we will extend this to prove convergence results for multiple
steps with random pairs of qubits drawn at every step.  We consider first the Haar distribution over the full unitary group
and then will discuss the more general case of any 2-copy gapped
distribution. 

In this section, we work in general dimension $d$ and with a general
Hermitian orthogonal basis $\sigma_0,\ldots,\sigma_{d^2-1}$.  Later we
will take $d$ to be either 4 or $2^n$ and the $\sigma_i$ to be Pauli
matrices.  However, in this section we keep the discussion general to
emphasise the potentially broader applications.

Fix an orthonormal basis for $d\times d$ Hermitian matrices:
$\sigma_0, \ldots, \sigma_{d^2-1}$, normalised so that $\tr
\sigma_p\sigma_q = d\,\delta_{p,q}$.   Let $\sigma_0$ be the identity.
We need to evaluate the quantity 
\begin{equation}
\label{eq:Tk}
\Expect_U \left(U^{\otimes k} \sigma_{p_1} \otimes \ldots \otimes \sigma_{p_k} (U^\dagger)^{\otimes k}\right) =: T(\bfp)
\end{equation}
where the expectation is over Haar distributed $U \in U(d)$.  We will need this quantity in two cases.  Firstly, for $d=2^n$, these are the moments obtained after applying a uniformly distributed unitary so we know what the random circuit must converge to.  Secondly, for $d=4$, this tells us how a random $U(4)$ gate acts on any chosen pair.

Call the quantity in \eq{Tk} $T(\bfp)$ (we use \textbf{bold} to indicate a $k$-tuple of coefficients; take $\bfp = (p_1, \ldots, p_k)$) and write it in the $\sigma_p$ basis as
\begin{equation}
T(\bfp) = \sum_{\bfq} \hat{G}(\bfq; \bfp) \sigma_{q_1} \otimes \ldots \otimes \sigma_{q_k}.
\end{equation}
Here, $\hat{G}(\bfq; \bfp)$ is the coefficient in the Pauli expansion of $T(\bfp)$ and we define $\hat{G}$ as the matrix with entries equal to $\hat{G}(\bfq; \bfp)$.  We have left off the usual normalisation factor because, as we shall see, with this normalisation $\hat{G}$ is a projector.  Inverting this, we have
\begin{align}
\hat{G}(\bfq; \bfp) 
&= d^{-k} \tr \left(\sigma_{q_1} \otimes \ldots \otimes \sigma_{q_k}
  T(\bfp)\right) \nonumber\\
& = d^{-k} \bbE_U \tr \left(  (\sigma_{q_1}\ot \cdots \ot \sigma_{q_k}) 
U^{\ot  k}  (\sigma_{p_1}\ot \cdots \ot \sigma_{p_k}) (U^\dag)^{\ot k}\right)
\label{eq:G}\end{align}
Note that $\hat{G}$ is real since $T$ and the basis are Hermitian.

We can gain all the information we need about the Haar integral in \eq{Tk} with the following observations:
\begin{lemma}
\label{lem:HaarIntegralCommutes}
$T(\bfp)$ commutes with $U^{\otimes k}$ for any unitary $U$.
\end{lemma}
\begin{proof}
Follows from the invariance of the Haar measure on the unitary group.
\end{proof}

\begin{corollary}
\label{cor:IsPerms}
$T(\bfp)$ is a linear combination of permutations from the symmetric group $S_k$.
\end{corollary}
\begin{proof}
This follows from Schur-Weyl duality (see e.g.~\cite{GoodmanWallach98}).
\end{proof}

From this, we can prove that $\hat{G}$ is a projector and find its eigenvectors.
\begin{theorem}
\label{thm:SymmetryG}
$\hat{G}$ is symmetric, i.e. $\hat{G}(\bfq; \bfp) = \hat{G}(\bfp; \bfq)$.
\end{theorem}
\begin{proof}
Follows from the invariance of the trace under cyclic permutations.
\end{proof}

\begin{theorem}
\label{thm:EigenvectorsG}
$P_\pi$ is an eigenvector of $\hat{G}$ with eigenvalue $1$ for any permutation operator $P_\pi$ i.e.
\begin{equation*}
\sum_{\bfq} \hat{G}(\bfp; \bfq) \tr (\sigma_{q_1} \otimes \ldots \otimes \sigma_{q_k} P_\pi) = \tr (\sigma_{p_1} \otimes \ldots \otimes \sigma_{p_k} P_\pi).
\end{equation*}
Further, any vector orthogonal to this set has eigenvalue $0$.
\end{theorem}
\begin{proof}
For the first part,
\begin{align}
\sum_{\bfq} \hat{G}(\bfp; \bfq) &
\tr (\sigma_{q_1} \otimes \ldots \otimes \sigma_{q_k} P_\pi) \nn
&= d^{-k} \sum_{\bfq} \mathbb{E}_U \tr \left(\sigma_{q_1} U \sigma_{p_1} U^\dagger\right) \ldots 
\tr \l(\sigma_{q_k} U
  \sigma_{p_k} U^\dagger\right) \tr\left(\sigma_{q_1} \otimes \ldots
  \otimes \sigma_{q_k} P_\pi \right)\nn
&= d^{-k} \tr \left( P_\pi \mathbb{E}_U \sum_{q_1} \tr \left(\sigma_{q_1} U
  \sigma_{p_1} U^\dagger\right) \sigma_{q_1} \otimes \ldots 
\otimes \sum_{q_k} \tr
\left(\sigma_{q_k} U \sigma_{p_k} U^\dagger\right) \sigma_{q_k} \r)
\label{eq:G-perm-evector}
\end{align}
Writing $U^\dagger \sigma_p U$ in the $\sigma_p$ basis, we find
\begin{equation*}
\frac{1}{d} \sum_q \tr\left( \sigma_q U \sigma_p U^\dagger \right) \sigma_q = U \sigma_p U^\dagger.
\end{equation*}
Therefore \eq{G-perm-evector} becomes %(using \lemref{TraceCycles})
\begin{equation*}
\tr \left( P_\pi \mathbb{E}_U U^\dagger \sigma_{p_1} U \otimes \ldots \otimes U^\dagger \sigma_{p_k} U \right) =\tr \left( \sigma_{p_1} \otimes \hdots \otimes \sigma_{p_k} P_\pi \right).
\end{equation*}
For the second part, consider any vector $v$ which is orthogonal to the permutation operators (we can neglect the complex conjugate because $P_{\pi}$ is real in this basis), i.e.
\begin{equation}
\label{eq:OrthToPerms}
\sum_{\bfq} \tr \left( \sigma_{q_1} \otimes \ldots \otimes \sigma_{q_k} P_{\pi} \right) v(\bfq) = 0
\end{equation}
for any permutation $\pi$.  Then
\begin{equation*}
\sum_{\bfq} \hat{G}(\bfp; \bfq) v(\bfq) = d^{-k} \sum_{\bfq} \tr \left(\sigma_{q_1} \otimes \ldots \otimes \sigma_{q_k} T(\bfp) \right) v(\bfq)
\end{equation*}
which is zero since $T(\bfp)$ is a linear combination of permutations and $v$ is orthogonal to this by \eq{OrthToPerms}.
\end{proof}

\begin{theorem}
$\hat{G}^2 = \hat{G}$, i.e. $\sum_{\mathbf{q'}} \hat{G}(\bfp; \mathbf{q'}) \hat{G}(\mathbf{q'}; \bfq) = \hat{G}(\bfp; \bfq)$.
\end{theorem}
\begin{proof}
Using \eq{G},
\begin{equation*}
\sum_{\mathbf{q'}} \hat{G}(\bfp; \mathbf{q'}) \hat{G}(\mathbf{q'}; \bfq) = \sum_{\mathbf{q'}} \hat{G}(\bfp; \mathbf{q'}) d^{-k} \tr \left(\sigma_{q'_1} \otimes \ldots \otimes \sigma_{q'_k} T(\bfq)\right).
\end{equation*}
From \corref{IsPerms}, $T(\bfq)$ is a linear combination of permutations.  This implies, using \thmref{EigenvectorsG} that
\begin{align*}
\sum_{\mathbf{q'}} \hat{G}(\bfp; \mathbf{q'}) d^{-k} \tr \left(\sigma_{q'_1} \otimes \ldots \otimes \sigma_{q'_k} T(\bfq)\right) &= d^{-k} \tr \left(\sigma_{p_1} \otimes \ldots \otimes \sigma_{p_k} T(\bfq)\right) \\
&= \hat{G}(\bfp; \bfq)
\end{align*}
as required.
\end{proof}

\begin{corollary}
$\hat{G}$ is a projector so has eigenvalues $0$ and $1$.
\end{corollary}

We now evaluate $\hat{G}$ and $T$ for the cases of $k=1$ and $k=2$ since these are the cases we are interested in for the remainder of the paper.

\subsection{\texorpdfstring{$k=1$}{k=1}}

The $k=1$ case is clear: the random unitary completely randomises the state.  Therefore all terms in the expansion are set to zero apart from the identity i.e.
\begin{equation}
T(p) =
\begin{cases}
\sigma_0	&	p=0\\
0	&	p \ne 0.
\end{cases}
\end{equation}

\subsection{\texorpdfstring{$k=2$}{k=2}}
\label{sec:GforK2}

For $k=2$, there are just two permutation operators, identity $I$ and swap $\mathcal{F}$.  Therefore there are just two eigenvectors with non-zero eigenvalue ($n > 1$).  In normalised form, taking them to be orthogonal, their components are
\begin{align*}
f_1(q_1, q_2) &= \delta_{q_1 0} \delta_{q_2 0} \\
f_2(q_1, q_2) &= \frac{1}{d^2-1} \delta_{q_1 q_2} (1 - \delta_{q_1 0})
\end{align*}
We will now prove three properties of $\hat{G}$ that we need:
\begin{enumerate}
\item{$\hat{G}(p_1, p_2; q_1, q_2) = 0$ if $p_1 \ne p_2$ or $q_1 \ne q_2$.
\begin{proof}
Consider the function $f(q_1, q_2) = \delta_{q_1 a} \delta_{q_2 b}$ with $a \ne b$.  This function has zero overlap with the eigenvectors $f_1$ and $f_2$ so it goes to zero when acted on by $\hat{G}$.  Therefore $\hat{G}(p_1, p_2; a, b) = 0$.  The claim follows from the symmetry property (\thmref{SymmetryG}).
\end{proof}}
With this we will write $\hat{G}(p; q) \equiv \hat{G}(p_1, p_2; q_1, q_2)$.
\item{$\hat{G}(p; 0) = \delta_{p 0}$.
\begin{proof}
Let $\hat{G}$ act on eigenvector $f_1$.
\end{proof}}
\item{$\hat{G}(p; a) = \frac{1}{d^2-1}$ for $a, p \ne 0$.
\begin{proof}
Let $\hat{G}$ act on the input $\delta_{q a}$.  This has zero overlap with $f_1$ and overlap $\frac{1}{d^2-1}$ with $f_2$.
\end{proof}}
\end{enumerate}
Therefore we have
\begin{equation}
\hat{G}(p_1, p_2; q_1, q_2) =
\begin{cases}
0	& p_1 \ne p_2 \rm{\, or\,} q_1 \ne q_2 \\
1	& p_1 = p_2 = q_1 = q_2 = 0 \\
\frac{1}{d^2-1}	& p_1 = p_2 \ne 0, q_1 = q_2 \ne 0 \\
\end{cases}
\end{equation}

Since $T(p_1, p_2) = \sum_{q_1, q_2} \hat{G}(p_1, p_2; q_1, q_2) \sigma_{q_1} \otimes \sigma_{q_2}$, we have
\begin{equation}
T(p_1, p_2) =
\begin{cases}
0	&	p_1 \ne p_2 \\
\sigma_0 \otimes \sigma_0	&	p_1 = p_2 = 0 \\
\frac{1}{d^2-1} \sum_{p' \ne 0} \sigma_{p'} \otimes \sigma_{p'}	&	p_1 = p_2 \ne 0.
\end{cases}
\end{equation}
Therefore the terms $\sigma_{p_1} \otimes \sigma_{p_2}$ with $p_1 \ne p_2$ are set to zero.  Further, the sum of the diagonal coefficients $\gamma(p, p)$ is conserved.  This allows us to identify this with a probability distribution (after renormalising) and use Markov chain analysis.  To see this, write again the starting state
\begin{equation*}
\rho = \frac{1}{d} \sum_{q_1, q_2} \gamma_0(q_1, q_2) \sigma_{q_1} \otimes \sigma_{q_2}
\end{equation*}
with state after application of any unitary $W$
\begin{equation*}
\rho_W = \frac{1}{d} \sum_{q_1, q_2} \gamma_W(q_1, q_2) \sigma_{q_1} \otimes \sigma_{q_2} = 2^{-n} \sum_{q_1, q_2} \gamma(q_1, q_2) \left(W \sigma_{q_1} W^{\dagger}\right) \otimes \left(W \sigma_{q_2} W^{\dagger}\right).
\end{equation*}
Then
\begin{align*}
\sum_q \gamma_W(q, q) &= \frac{1}{d} \sum_q \tr \left(\sigma_q \otimes \sigma_q \rho_W\right) \\
&= \tr \left(\swap \rho_W\right) \\
&= \frac{1}{d} \sum_{q_1, q_2} \gamma(q_1, q_2) \tr \left( \swap \left(W \sigma_{q_1} W^{\dagger}\right) \otimes \left(W \sigma_{q_2} W^{\dagger}\right) \right)\\
&= \frac{1}{d} \sum_{q_1, q_2} \gamma(q_1, q_2) \tr \left( \sigma_{q_1} \sigma_{q_2} \right)\\
&= \sum_q \gamma(q, q)
\end{align*}
as required, where $\swap$ is the swap operator and we have used
Lemmas \ref{lem:Swap} and \ref{lem:TraceCycles}. 

\subsection{Moments for General Universal Random Circuits}
\label{sec:MomentsGeneral}

We now consider universal distributions $\mu$ that in general may be
different from the uniform (Haar) measure on $U(d)$.  Our main result in this section
will be to show that a universal distribution on $U(4)$ is also 2-copy
gapped.  In fact, we will phrase this result in slightly more general
terms and show that a universal distribution on $U(d)$ is also
$k$-copy gapped for any $k$.  Universality (\defref{dist-universal}) generalises in
the obvious way to $U(d)$, whereas when we say that $\mu$ is $k$-copy
gapped, we mean that \be \|G_\mu - G_{U(d)}\|_\infty < 1,
\label{eq:k-copy-gapped}\ee
where $G_{?}=\bbE_U U^{\ot k} \ot (U^*)^{\ot k}$, with the
expectation taken over $\mu$ for $G_\mu$ or over the Haar measure for
$G_{U(d)}$.  

The reason \eq{k-copy-gapped} represents our condition for $\mu$ to be
$k$-copy gapped is as follows:
Observe that $\hat{G}$ and $G$ are unitarily related, so the
definition of $k$-copy gapped could equivalently be given in terms of
$\hat{G}$.  We have shown above that $\hat{G}_{U(d)}$ (and thus
$G_{U(d)}$) has all eigenvalues equal to $0$ or $1$; i.e. is a projector.  By
contrast, $G_\mu$ may not even be Hermitian.  However, we will prove
below that all eigenvectors of ${G}_{U(d)}$ with eigenvalue 1 are also
eigenvectors of ${G}_\mu$ with eigenvalue 1.  Thus, \eq{k-copy-gapped}
will imply that $\lim_{t\ra\infty} (\hat{G}_\mu)^t= \hat{G}_{U(d)}$,
just as we would expect for a gapped random walk.

 We would like to show that \eq{k-copy-gapped} holds whenever
$\mu$ is universal.  This result was proved in \cite{ArnoldKrylov62}
(and was probably known even earlier)
when $\mu$ had the form $(\delta_{U_1}+\delta_{U_2})/2$.  Here we show how
to extend the argument to any universal $\mu$.

\begin{lemma}
\label{lem:GEigenvalues}
Let $\mu$ be a distribution on $U(d)$. Then all eigenvectors of $G_{U(d)}$ with
eigenvalue 1 are eigenvectors of $G_\mu$ with eigenvalue one.
Additionally, if $\mu$ is universal then $\mu$ is $k$-copy gapped for
any positive integer $k$ (cf. \eq{k-copy-gapped}).
\end{lemma}
In particular, if $k=2$ this Lemma implies that $\mu$ is 2-copy gapped
(cf. \thmref{kCopyGappedExamples}).
\begin{proof}
Let $V\cong \bbC^{d}$ be the fundamental representation of $U(d)$,
where the action of $U\in U(d)$ is simply $U$ itself.  Let $V^*$ be
its dual representation, where $U$ acts as $U^*$.  The operators
$G_\mu$ and $G_{U(d)}$ act on the space $V^{\ot k} \ot (V^*)^{\ot
  k}$.  We will see that $G_{U(d)}$ is completely determined by the
decomposition of $V^{\ot k} \ot (V^*)^{\ot k}$ into irreducible representations (irreps).  
 Suppose that
the multiplicity of $(r_\lambda,V_\lambda)$ in $V^{\ot k} \ot
(V^*)^{\ot k}$ is $m_\lambda$,
where the $V_\lambda$'s are the irrep spaces and $r_\lambda(U)$ the
corresponding representation matrices.  In other words
\begin{align}
V^{\ot k} \ot (V^*)^{\ot k} 
& \cong \bigoplus_{\lambda} V_\lambda \ot \bbC^{m_\lambda} \\
U^{\ot k} \ot (U^*)^{\ot k}
& \sim \sum_{\lambda} \proj{\lambda} \ot r_\lambda(U) \ot
I_{m_\lambda}
\label{eq:irrep-basis}
\end{align}
Here $\sim$ indicates that the two sides are related by conjugation by
a fixed ($U$-independent) unitary.

Let
$\lambda=0$ denote the trivial irrep: i.e. $V_0=\bbC$ and $r_0(U)=1$
for all $U$. 
We claim that $\bbE_U r_\lambda(U)=0$ whenever 
$\lambda\neq 0$ and the expectation is taken over the Haar measure.
To show this, note that $\bbE_U
r_\lambda(U)$ commutes with $r_\lambda(V)$ for all $V\in U(d)$ and
thus, by Schur's Lemma, we must have $\bbE_U
r_\lambda(U)= cI$ for some $c\in \bbC$.  However, by the
translation-invariance of 
the Haar measure we have $cI = \bbE_U
r_\lambda(U) = \bbE_U
r_\lambda(UV) = c\, r_\lambda(V)$ for all $V\in U(d)$.  Since
$\lambda\neq 0$, we cannot have $r_\lambda(V)= I$ for all $V$ and so
it must be that $c=0$.

Thus, if we write $G_{U(d)}$ and  $G_\mu$ using the basis on the RHS of
\eq{irrep-basis}, we have
\be {G}_{U(d)}  = \proj{0} \ot I_{m_0}\ee
where $\proj{0}$ is a projector onto the trivial irrep.
On the other hand,
\be {G}_{\mu} = \proj{0} \ot I_{m_0} + \sum_{\lambda \ne 0} \proj{\lambda} \ot \l(\int r_\lambda(U) d\mu(U)\r) \ot I_{m_\lambda}
\ee
Thus, every eigenvector of ${G}_{U(d)}$ with eigenvalue one is also
fixed by ${G}_\mu$.  For the remainder of the space, the direct sum
structure means that
\be \|G_{U(d)} - G_\mu\|_\infty = \max_{\substack{\lambda\neq 0\\
m_\lambda\neq 0}} \l\|\int r_\lambda(U) d\mu(U)\r\|_\infty
.\ee
Note that this maximisation only includes $\lambda$ with $\dim V_\lambda>1$.  This is because non-trivial one-dimensional irreps of $U(d)$ have the form $\det U^m$ for some non-zero integer $m$.  Under the map $U\mapsto e^{i\phi}U$, such irreps pick up a phase of $e^{im\phi}$.  However, $U^{\ot k}\ot (U^*)^{\ot k}$ is invariant under $U\mapsto e^{i\phi}U$.  Thus $V^{\ot k} \ot (V^*)^{\ot k}$  cannot contain any non-trivial one-dimensional irreps.

Now suppose by
contradiction that there exists $\lambda\neq 0$ 
with $m_\lambda\neq 0$
and $\|\int r_\lambda(U) d\mu(U)\|_\infty=1$.  (We do not need to consider the case $\|\int
r_\lambda(U) d\mu(U)\|_\infty>1$, since
$\|r_\lambda(U)\|_\infty=1$ for all $U$ and $\|\cdot\|_\infty$ obeys the triangle inequality.)
Indeed, the triangle inequality further implies that  there exists a unit vector $\ket{v}\in
V_\lambda$ such that  
$$\int d\mu(U)\, r_\lambda(U)\ket{v} = \omega\ket{v},$$
for some $\omega\in\bbC$ with $|\omega|=1$. 

By the above argument we can assume that $\dim V_\lambda > 1$.  Since $V_\lambda$ is irreducible, it cannot contain a one-dimensional invariant subspace, implying that there exists $U_0\in U(d)$ such that 
$$|\bra{v}r_\lambda(U_0)\ket{v}| = 1-\delta,$$
for some $\delta>0$.   Since $U\mapsto |\bra{v}r_\lambda(U)\ket{v}|$ is
continuous, there exists an open ball $S$ around $U_0$ such that
$|\bra{v}r_\lambda(U)\ket{v}| \leq  1-\delta/2$ for all $U\in S$.  
Define $\bar{S} := U(d)\backslash S$.

Now we use the fact that $\mu$ is universal to find an $\ell$ such
that $\mu^{\star \ell}(S)>0$.  Next, observe that $\int d\mu^{\star
  \ell}(U)\, \bra{v}r_\lambda(U)\ket{v} = \omega^\ell$.  Taking the
absolute value of both sides yields
\begin{align*}
1 & = 
\l| \int_{U(d)} d\mu^{\star \ell}(U)\, \bra{v}r_\lambda(U)\ket{v}\r|
\\ & \leq
 \int_{U(d)} d\mu^{\star \ell}(U)\, \l|\bra{v}r_\lambda(U)\ket{v}\r|
\\ & =
 \int_{S} d\mu^{\star \ell}(U)\, \l|\bra{v}r_\lambda(U)\ket{v}\r|
+ \int_{\bar{S}} d\mu^{\star \ell}(U)\, \l|\bra{v}r_\lambda(U)\ket{v}\r|
\\ & \leq
\mu^{\star \ell}(S)\l(1-\frac{\delta}{2}\r) + 
\l(1-\mu^{\star \ell}(S)\r)
\\ & < 1,
\end{align*}
a contradiction.  We conclude that $\|G_{U(d)} - G_\mu\|_\infty<1$.
\end{proof}

%For our purposes, we are usually interested in the case when $d=4$ so
%that we can compare a random choice from a $U(4)$-universal gate set
%with picking a uniformly random element of $U(4)$.
%\lemref{GEigenvalues} then implies that $\hat{G}_\mu$ has a nonzero
%gap for any universal choice of $\mu$.  When we construct random
%circuits on $n$ qubits, our choice of distribution $\mu$ will affect
%the gap of the resulting Markov chain (described in the next section)
%only by an $n$-independent multiplicative factor.

%In particular, it will not change the scaling of the gap with $n$,
%although it is possible that the uniform ensemble mixes faster.

%This tells us that our random circuits converge and that for gate sets
%that are locally universal, random circuits on just one pair converge
%as well.  Therefore there is an eigenvalue gap between $\hat{G}^{U(4)}$ and
%$\hat{G}^S$ (the dimension is $4$) which does not depend on $n$.

%We now turn to proving that $\hat{G}^S$ converges to
% $\hat{G}^{U(2^n)}$ quickly.  A large eigenvalue gap will mean the
% eigenvalues not equal to $1$ will decay quickly so bounding the gap
% will give a measure of the convergence rate.  We solve this for
% $k=1, 2$.

\section{Convergence}
\label{sec:Convergence}

In \secref{Moments} we saw that iterating any universal gate set on $U(d)$ eventually converges to the uniform distribution on $U(d)$.  Since the set of all two-qubit unitaries is universal on $U(2^n)$, this implies that random circuits eventually converge to the Haar measure.  In this section, we turn to proving upper bounds on this convergence rate, focusing on the first two moments.

Let $\hat{G}^{(ij)}$ be the matrix with $\hat{G}$ (with $d=4$) acting on qubits $i$ and $j$ and the identity on the others.  Then, if the pair $(i, j)$ is chosen at step $t$, we can find the coefficients at step $t+1$ by multiplying by $\hat{G}^{(ij)}$.  In general, a random pair is chosen at each step.  So
\be
\gamma_{t+1}(\mathbf{p}) = \sum_{\mathbf{q}} \frac{1}{n(n-1)} \sum_{i \ne j} \hat{G}^{(ij)}(\mathbf{p}; \mathbf{q}) \gamma_t(\mathbf{q})
\ee
where $\gamma_{t+1}$ are the expected coefficients at step $t$.  We can think of this evolution as repeated application of the matrix
\begin{equation}
\label{eq:GeneralTransitionMatrix}
P = \frac{1}{n(n-1)} \sum_{i \ne j} \hat{G}^{(ij)}.
\end{equation}

For $k=2$, the key idea of Oliveira et al.~\cite{ODP06} was to map the evolution of the $\gamma(p, p)$ coefficients to a Markov chain.  The $\gamma(p_1, p_2)$ coefficients with $p_1 \ne p_2$ just decay as each qubit is chosen and can be analysed directly.

However, we can only map the $\gamma(p, p)$ coefficients to a probability distribution when they are non-negative, which is not the case for general states.  Most of the rest of the paper is dedicated to proving \lemref{MainMixing}, which only applies to states with $\gamma(p, p) \ge 0$ and normalised so their sum is $1$.  \corref{MainMixing} then extends this to all states:
\begin{proof}[of \corref{MainMixing}]
\lemref{MainMixing} still applies to the $\gamma(p_1, p_2)$ terms with $p_1 \ne p_2$.  Therefore we just need to show how to apply \lemref{MainMixing} to states that initially have some negative $\gamma(p, p)$ terms.  

For the $\gamma(p, p)$ terms, \lemref{MainMixing} says that the random walk starting with any initial probability distribution converges to uniform in some bounded time $t$.  Let $g_t(p, p; q, q)$ be the coefficients after $t$ steps of the walk starting at a particular point $q$ (i.e.~$g_0(p, p;q,q) = \delta_{p,q}$).  Now, for any starting state $\rho$, let the initial coefficients be $\gamma_0(p,p)$.  Then, by linearity, we can write the expected coefficients after $t$ steps $\gamma_t(p, p) := \Expect \gamma_W(p, p)$ as
\be
\gamma_t(p,p) = \sum_{q\ne0} \gamma_0(q,q) g_t(p,p;q,q)
\ee
for $p\ne0$.

We can now prove convergence rates for the expected coefficients $\gamma_t(p, p)$:
\begin{itemize}
\item[(i)]{
For the 2-norm, we have from \lemref{MainMixing} that for $t \ge C n \log 1/\eps$
\be
\sum_{p \ne 0} \left( g_t(p, p; q, q) - \frac{1}{4^n-1}\right)^2 \le \eps
\ee
for any $q$.  Note that the normalisation for the $\gamma(p,p)$ terms with $p\ne0$ has changed from \lemref{MainMixing} since we are neglecting the $\gamma(0,0)$ term here.  Now
\begin{align*}
&\sum_{p \ne 0} \left( {\gamma}_t(p,p) - \frac{\sum_{q\ne0} \gamma_0(q,q)}{4^n-1}\right)^2 \\
&= \sum_{p \ne 0} \left( \sum_{q \ne 0} \gamma_0(q,q) \left(g_t(p,p; q,q) - \frac{1}{4^n-1}\right)\right)^2 \\
&\le \sum_{q \ne 0} \gamma_0(q,q)^2 \sum_{q' \ne 0} \sum_{p\ne0} \left(g_t(p,p;q',q') - \frac{1}{4^n-1}\right)^2 \\
&\le (4^n-1) \eps \sum_{q \ne 0} \gamma_0(q,q)^2 \\
&\le 4^n \eps \sum_{q_1, q_2} \gamma_0(q_1,q_2)^2 \\
&= 4^n \eps \, \tr \rho^2 \\
&\le 4^n \eps
\end{align*}
where the first inequality is the Cauchy-Schwarz inequality.  Therefore for $t \ge C n (n + \log 4^n/\eps)$, the 2-norm distance from stationarity for the $\gamma(p,p)$ terms is at most $\eps$.  Choose $C'$ such that $C' n (n + \log 1/\eps) \ge C n (n + \log 4^n/\eps)$ to obtain the result.
}
\item[(ii)]{
For the 1-norm, \lemref{MainMixing} says that for $t \ge C n (n + \log 1/\eps)$
\be
\sum_{p \ne 0} \left| g_t(q;p,p) - \frac{1}{4^n-1}\right| \le \eps.
\ee
We can then proceed much as for the 2-norm case:
\begin{align*}
&\sum_{p \ne 0} \left| \gamma_t(p,p) - \frac{\sum_{q\ne0} \gamma_0(q,q)}{4^n-1}\right| \\
&= \sum_{p \ne 0} \left| \sum_{q \ne 0} \gamma_0(q,q) \left(g_t(p,p;q,q) - \frac{1}{4^n-1}\right)\right| \\
&\le \sum_{q \ne 0} |\gamma_0(q,q)| \sum_{p\ne0} \left|g_t(p,p;q,q) - \frac{1}{4^n-1}\right| \\
&\le \eps \sum_{q \ne 0} |\gamma_0(q,q)| \\
&\le 2^n \eps.
\end{align*}
The last inequality follows from $|\sigma_q \otimes \sigma_q| = \sigma_0 \otimes \sigma_0$.  Therefore for $t \ge C n (n + \log 2^n/\eps)$, the 1-norm distance from stationarity for the $\gamma(p,p)$ terms is at most $\eps$.
}\qedhere
\end{itemize}
\end{proof}

We now proceed to prove \lemref{MainMixing}.  Firstly, we will consider the simple case of $k=1$ to prove this process forms a 1-design as this will help us to understand the more complicated case of $k=2$.

\subsection{First Moments Convergence}

Recall that $\rho = 2^{-n/2} \sum_p \gamma(p) \sigma_p$ and we wish to evaluate the moments of the coefficients.  So for the first moments to converge, we want to know $\Expect \gamma(p)$.

For $k=1$, the $U(4)$ random circuit uniformly randomises each pair that is chosen.  More precisely, a pair of sites $i, j$ are chosen at random and all the coefficients with $p_i \ne 0$ or $p_j \ne 0$ are set to zero.  Thus we get an exact 1-design when all sites have been hit.  For other gate sets, the terms do not decay to zero but decay by a factor depending on the gap of $\hat{G}$.  Call the gap $\Delta$; for $U(4)$ $\Delta=1$ and for others $0 < \Delta \le 1$ and $\Delta$ is independent of $n$.  Therefore once each site has been hit $m$ times the terms have decayed by a factor $(1-\Delta)^m$.

For a bound like the mixing time (see \secref{MarkovChain} for definition), we want to bound the quantity $\sum_{p \ne 0} | \Expect_W \gamma_W(p) |$ where $\gamma_W(p)$ is the Pauli coefficient after applying the random circuit $W$.  We also want 2-norm bounds, so we bound $\sum_{p \ne 0} (\Expect_W \gamma_W(p))^2$ too.  We will in fact find bounds on $\sum_{p \ne 0} \Expect_W | \gamma_W(p) |$ and $\sum_{p \ne 0} (\Expect_W |\gamma_W(p)|)^2$, which are stronger.

A standard problem in the theory of randomised algorithms is the `coupon collector' problem.  If a magazine comes with a free coupon, which is chosen uniformly randomly from $n$ different types, how many magazines should you buy to have a high probability of getting all $n$ coupons?  It is not hard to show that $n \ln \frac{n}{\eps}$ samples (magazines) have at least a $1-\eps$ probability of including all $n$ coupons.  Using this, we expect all sites to be hit with probability at least $1-\eps$ after $\Theta(n \log \frac{n}{\eps})$ steps.  This argument can be made precise in this context by bounding the non-identity coefficients.  We find, as expected, that the sum is small after $O(n \log n)$ steps:
\begin{lemma}
\label{lem:CoefficientDecayGeneral}
After $O(n \log 1/\eps)$ steps
\begin{equation*}
\sum_{p \ne 0} \left( \Expect_W |\gamma_W(p) | \right)^2 \le \eps
\end{equation*}
and after $O(n\log \frac{n}{\eps})$ steps,
\begin{equation}
\label{eq:TermsDecayGeneral}
\sum_{p \ne 0} \Expect_W | \gamma_W(p) | \le \eps.
\end{equation}
\end{lemma}
\begin{proof}
At each step, a pair of sites is chosen at random and any terms with non-identity coefficients for this pair decay by a factor $(1-\Delta)$.  For example, the term $\sigma_1 \otimes \sigma_0^{\otimes (n-1)}$ decays whenever the first site is chosen.  Thus the probability of each term decaying depends on the number of zeroes.  We start with the 1-norm bound.

Suppose the circuit applied after $t$ steps is $W_t$.  Consider $\Expect_{W_t} | \gamma_{W_t}(p) |$ for any $p$ with $d$ non-zeroes.  Since the state $\rho$ is physical, $\tr \rho^2 \le 1$ so $\sum_p \gamma^2_0(p)  \le 1$.  Now, in each step, if any site is chosen where $p$ is non-zero, this term decays by a factor $(1-\Delta)$.  This occurs with probability $1-\frac{(d-n)(d-n-1)}{n(n-1)} \ge d/n$, the probability of choosing a pair where at least one site is non-zero.  Therefore
\begin{equation*}
\Expect | \gamma_{W_t}(p) | \le \left( (1-\Delta)d/n + (1-d/n) \right) | \gamma_{W_{t-1}}(p) |
\end{equation*}
where the expectation is over the circuit applied at step $t$.  If we iterate this $t$ times we find
\begin{equation*}
\Expect_{W} | \gamma_{W}(p) | \le \exp(-\Delta t d/n) | \gamma_{0}(p) |
\end{equation*}
where the expectation here is over all random circuits for the $t$ steps.  We now sum over all $p$:
\begin{equation*}
\sum_{p \ne 0} \Expect_W | \gamma_W(p) | \le \sum_{d=1}^{n} \exp(-\Delta t d/n) \sum_{d(p) = d} | \gamma_{0}(p) |
\end{equation*}
where $d(p)$ is the number of non-zeroes in $p$.  For the 1-norm bound, we can simply bound $| \gamma_{0}(p) | \le 1$ to give $\sum_{d(p) = d} | \gamma_0(p) | \le {n \choose d} 3^d$ so
\begin{equation*}
\sum_{p \ne 0} \Expect_W | \gamma_W(p) | \le (1+3 \exp(-\Delta t /n))^n - 1
\end{equation*}
where we have used the binomial theorem.  Now let $t= \frac{n}{\Delta} \ln \frac{3n}{\eps}$.  This gives
\begin{equation*}
\sum_{p \ne 0} \Expect_W | \gamma_W(p) | \le (1+\eps/n)^n - 1 = O(\eps).
\end{equation*}
For the 2-norm bound,
\begin{align*}
\sum_{p \ne 0} (\Expect_W | \gamma_W(p) |)^2 \le& \sum_{p \ne 0} \exp(-2\Delta t d/n) \gamma^2_{0}(p) \\
=& \sum_{d=1}^n \exp(-2\Delta t d/n)\sum_{d(p) = d} \gamma^2_{0}(p) \\
\le& \sum_{d=1}^n \exp(-2\Delta t d/n) \\
\le& \frac{\exp(-2\Delta t /n)}{1-\exp(-2\Delta t /n)}
\end{align*}
where we have used $\sum_p \gamma^2_0(p) \le 1$.  We find after $\frac{n}{2\Delta} \ln 1/\eps$ steps that
\begin{equation*}
\sum_{p \ne 0} (\Expect_W |\gamma_W(p)|)^2 \le \frac{\eps}{1-\eps}\qedhere
\end{equation*}
\end{proof}

\subsection{Second Moments Convergence}

Firstly, the $\sigma_{p_1} \otimes \sigma_{p_2}$ terms for $p_1 \ne p_2$ decay in a similar way to the non-identity terms in the 1-design analysis.  In fact, the proof of \lemref{CoefficientDecayGeneral} carries over almost identically to this case to give
\begin{lemma}
\label{lem:CoefficientsDecay2General}
After $O(n \log 1/\eps)$ steps
\begin{equation*}
\sum_{p_1 \ne p_2} (\Expect_W |\gamma_W(p_1, p_2)|)^2 \le \eps
\end{equation*}
and after $O(n(n+\log 1/\eps))$ steps
\begin{equation*}
\sum_{p_1 \ne p_2} \Expect_W | \gamma_W(p_1, p_2) | \le \eps.
\end{equation*}
\end{lemma}
\begin{proof}
Instead of the number of zeroes governing the decay rate, we need to count the number of places where $p_1$ and $p_2$ differ.  This gives
\begin{equation*}
\Expect | \gamma_{W_t}(p_1, p_2) | \le \left( (1-\Delta)d/n + (1-d/n) \right) | \gamma_{W_{t-1}}(p_1, p_2) |
\end{equation*}
where now $d$ is the number of differing sites.  There are ${n \choose d} 12^d 4^{n-d}$ states that differ in $d$ places so we find
\begin{equation*}
\sum_{p_1 \ne p_2} \Expect_W | \gamma_W(p_1, p_2) | \le 4^n[(1+3 \exp(-\Delta t /n))^n - 1].
\end{equation*}
Set $t =\frac{n}{\Delta} (n \ln 4 + \ln 1/\eps)$ to make this $O(\eps)$.  The 2-norm bound follows in the same way as for \lemref{CoefficientDecayGeneral}.
\end{proof}
We now need to prove the $\gamma(p, p)$ terms converge quickly.  We have seen above that the sum of the terms $\gamma(p, p)$ is conserved and, for the purposes of proving \lemref{MainMixing}, we assume the sum is $1$ and $\gamma(p,p) \ge 0$ for all $p$.

To illustrate the evolution, consider the simplest case when the gates are chosen from $U(4)$.  We have evaluated $\hat{G}$ in \secref{GforK2} for $k=2$ for this case.  Translated into coefficients this yields the following update rule, where we have written it for the case when qubits 1 and 2 are chosen:
\begin{multline}
\label{eq:FullChain}
\gamma_{t+1}(r_1, r_2, r_3, \ldots, r_n, s_1, s_2, s_3, \ldots, s_n) \\
=\begin{cases}
0	&	(r_1, r_2) \ne (s_1, s_2) \\
\gamma_{t}(0, 0, r_3, \ldots, r_n, 0, 0, s_3, \ldots, s_n)	&	(r_1, r_2) = (s_1, s_2) = (0, 0) \\
\frac{1}{15} \sum_{r'_1, r'_2 \atop r'_1 r'_2 \ne 0} \gamma_{t}(r'_1, r'_2, r_3, \ldots, r_n, r'_1, r'_2, s_3, \ldots, s_n) &	(r_1, r_2) = (s_1, s_2) \ne (0, 0).
\end{cases}
\end{multline}
The key idea of Oliveira et al.~\cite{ODP06} was to map the evolution of the $\gamma(p, p)$ coefficients to a Markov chain.  We can apply this here to get, on state space $\{0, 1, 2, 3\}^n$, the evolution:
\begin{enumerate}
\item{Choose a pair of sites uniformly at random.}
\item{If the state is $00$ it remains $00$.}
\item{Otherwise, choose the state uniformly at random from $\{0,1,2,3\}^2 \backslash \{00\}$.}
\end{enumerate}
This is the correct evolution since, if the initial state is distributed according to $\gamma_t(q, q)$, the final state is distributed according to $\gamma_{t+1}(p, p)$.

The evolution for other gate sets will be similar, but the states will not be chosen uniformly randomly in the third step.  However, the state $00$ will remain $00$ and the stationary distribution on the other 15 states is the same.  We will find the convergence times for general gate sets and then consider the $U(4)$ gate set since we can perform a tight analysis for this case.

\subsection{Markov Chain Analysis}
\label{sec:MarkovChain}

Before finding the convergence rate for our problem, we will briefly introduce the basics of Markov chain mixing time analysis.  All of these standard results can be found in Ref.~\cite{MontenegroTetali06} and references therein.

A process is Markov if the evolution only depends on the current state rather than the full state history.  Therefore the evolution of the state can be thought of as a matrix, the \emph{transition matrix}, acting on a vector which represents the current distribution.  We will only be interested in discrete time processes so the state after $t$ steps is given by the $t^{\text{th}}$ power of the transition matrix acting on the initial distribution.

We say a Markov chain is \emph{irreducible} if it is possible to get from one state to any other state in some number of steps.  Further, a chain is \emph{aperiodic} if it does not return to a state at regular intervals.  If a chain is both irreducible and aperiodic then it is said to be \emph{ergodic}.  A well known result of Markov chain theory is that all ergodic chains converge to a unique stationary distribution.  In matrix language this says that the transition matrix $P$ has eigenvalue $1$ with no multiplicity and all other eigenvalues have absolute value strictly less than 1.  We will also need the notion of \emph{reversibility}.  A Markov chain is reversible if the time reversed chain has the same transition matrix, with respect to some distribution.  This condition is also known as \emph{detailed balance}:
\begin{equation}
\label{eq:DetailedBalance}
\pi(x) P(x, y) = \pi(y) P(y, x).
\end{equation}
It can be shown that a reversible ergodic Markov chain is only reversible with respect to the stationary distribution.  So above $\pi(x)$ is the stationary distribution of $P$.  An immediate consequence of this is that for a chain with uniform stationary distribution, it is reversible if and only if it is symmetric (i.e.~$P(x, y) = P(y, x)$).  Note also that reversible chains have real eigenvalues, since they are similar to the symmetric matrix $\sqrt{\frac{\pi(x)}{\pi(y)}}P(x, y)$.

With these definitions and concepts, we can now ask how quickly the Markov chain converges to the stationary distribution.  This is normally defined in terms of the 1-norm mixing time.  We use (half the) 1-norm distance to measure distances between distributions:
\begin{equation}
\vectornorm{s - t} = \frac{1}{2} \vectornorm{s - t}_1 = \frac{1}{2} \sum_i |s_i - t_i|.
\end{equation}
We assume all distributions are normalised so then $0 \le \vectornorm{s - t} \le 1$.  We can now define the mixing time:
\begin{definition}
Let $\pi$ be the stationary distribution of $P$.  Then if $P$ is ergodic the mixing time $\tau$ is
\begin{equation}
\tau(\eps) = \max_s \min_t \{ t \ge 0 : \vectornorm{P^t s - \pi} \le \eps \}.
\end{equation}
\end{definition}
We will also use the (weaker) 2-norm mixing time (note this is not the same as $\tau_2$ in Ref.~\cite{MontenegroTetali06}):
\begin{definition}
Let $\pi$ be the stationary distribution of $P$.  Then if $P$ is ergodic the 2-norm mixing time $\tau_2$ is
\begin{equation}
\tau_2(\eps) = \max_s \min_t \{ t \ge 0 : \vectornorm{P^t s - \pi}_2 \le \eps \}.
\end{equation}
\end{definition}
Unless otherwise stated, when we say mixing time we are referring to the 1-norm mixing time.

There are many techniques for bounding the mixing time, including finding the second largest eigenvalue of $P$.  This gives a good measure of the mixing time because components parallel to the second largest eigenvector decay the slowest.  We have (for reversible ergodic chains)
\begin{theorem}[see Ref.~\cite{MontenegroTetali06}, Corollary 1.15]
\label{thm:MixingTimeGap}
\begin{equation*}
\tau(\eps) \le \frac{1}{\Delta} \ln \frac{1}{\pi_* \eps}
\end{equation*}
where $\pi_* = \min \pi(x)$ and $\Delta = \min(1-\lambda_2, 1+\lambda_{min})$ where $\lambda_2$ is the second largest eigenvalue and $\lambda_{min}$ is the smallest.  $\Delta$ is known as the \emph{gap}.
\end{theorem}
If the chain is irreversible, it may not even have real eigenvalues.  However, we can bound the mixing time in terms of the eigenvalues of the reversible matrix $PP^*$ where $P^*(x, y) = \frac{\pi(y)}{\pi(x)} P(y, x)$.  In this case we have (\cite{MontenegroTetali06}, Corollary 1.14)
\begin{equation}
\tau(\eps) \le \frac{2}{\Delta_{PP^*}} \ln \frac{1}{\pi_* \eps}
\end{equation}
where now $\Delta_{PP^*}$ is the gap of the chain $PP^*$.  Note that for a reversible chain $P = P^*$ and $\Delta_{PP^*} \approx 2\Delta$ so the bounds are approximately the same.

This can also be converted into a 2-norm mixing time bound:
\begin{equation}
\label{eq:2normMixing}
\tau_2(\eps)\le \frac{2}{\Delta_{PP^*}} \ln 1/\eps.
\end{equation}
To bound the gap, we will use the comparison theorem in \thmref{Comparison} below.  In this Theorem, we are thinking of the Markov chain as a directed graph where the vertices are the states and there are edges for allowed transitions (i.e.~transitions with non-zero probability).  For irreducible chains, it is possible to make a path from any vertex to any other; we call the path length the number of transitions in such a path (which will in general depend on the choice of path).
\begin{theorem}[see Ref.~\cite{MontenegroTetali06}, Theorem 2.14]
\label{thm:Comparison}
Let $P$ and $\hat{P}$ be two Markov chains on the same state space $\Omega$ with the same stationary distribution $\pi$.  Then, for every $x \ne y \in \Omega$ with $\hat{P}(x, y) > 0$ define a directed path $\gamma_{xy}$ from $x$ to $y$ along edges in $P$ and let its length be $| \gamma_{xy} |$.  Let $\Gamma$ be the set of all such paths.  Then
\begin{equation*}
\Delta \ge \hat{\Delta}/A
\end{equation*}
for the gaps $\Delta$ and $\hat{\Delta}$ where
\begin{equation*}
A = A(\Gamma) = \max_{a \ne b, P(a, b) \ne 0} \frac{1}{\pi(a) P(a, b)} \sum_{x \ne y : (a, b) \in \gamma_{xy}} \pi(x) \hat{P}(x, y) | \gamma_{xy} |.
\end{equation*}
\end{theorem}
For example, when comparing 1-dimensional random walks there is no choice in the paths; they must pass through every point between $x$ and $y$.  Further, the walk can only progress one step at a time so (without loss of generality, for reversible chains) let $b = a+1$ to give
\begin{align}
\label{eq:ComparisonWalk}
A &= \max_a \frac{1}{\pi(a) P(a, a+1)} \sum_{x \le a} \sum_{y \ge a+1} \pi(x) \hat{P}(x, y) (y-x) \nonumber \\
&= \max_a \frac{\hat{P}(a, a+1)}{P(a, a+1)}.
\end{align}
A generalisation of the comparison theorem involves constructing flows, which are weighted sets of paths between states.  This can give a tighter bound since bottlenecks are averaged over.  This gives a modified comparison theorem:
\begin{theorem}[\cite{DiaconisSaloff-Coste93}, Theorem 2.3]
\label{thm:ComparisonFlows}
Let $P$ and $\hat{P}$ be two Markov chains on the same state space $\Omega$ with the same stationary distribution $\pi$.  Then, for every $x \ne y \in \Omega$ with $\hat{P}(x, y) > 0$, construct a set of directed paths $\mathcal{P}_{xy}$ from $x$ to $y$ along edges in $P$.  We define the flow function $f$ which maps each path $\gamma_{xy} \in \mathcal{P}_{xy}$ to a real number in the interval $[0, 1]$ such that
\begin{equation*}
\sum_{\gamma_{xy} \in \mathcal{P}_{xy}} f(\gamma_{xy}) = \hat{P}(x, y).
\end{equation*}
Again, let the length of each path be $| \gamma_{xy} |$.  Then
\begin{equation*}
\Delta \ge \hat{\Delta}/A
\end{equation*}
for the gaps $\Delta$ and $\hat{\Delta}$ where
\begin{equation}
\label{eq:ComparisonFlowsA}
A = A(f) = \max_{a \ne b, P(a, b) \ne 0} \frac{1}{\pi(a) P(a, b)} \sum_{x \ne y, \gamma_{xy} \in \mathcal{P}_{xy} : (a, b) \in \gamma_{xy}} \pi(x) f(\gamma_{xy}) | \gamma_{xy} |.
\end{equation}
\end{theorem}
Note that we recover the comparison theorem when there is just one path between each $x$ and $y$.

\subsubsection{log-Sobolev Constant}

We will need tighter, but more complicated, mixing time results to prove the tight result for the $U(4)$ case.  We use the log-Sobolev constant:
\begin{definition}
\label{def:LogSobolev}
The log-Sobolev constant $\rho$ of a chain with transition matrix $P$ and stationary distribution $\pi$ is
\begin{equation*}
\rho = \min_f \frac{\sum_{x \ne y} (f(x) - f(y))^2 P(x, y) \pi(y)}{\sum_x \pi(x) f(x)^2 \log \frac{f(x)^2}{\sum_y \pi(y) f(y)^2}}.
\end{equation*}
\end{definition}
The mixing time result is:
\begin{lemma}[see Ref.~\cite{DiaconisSaloff-Coste96}, Theorem 3.7']
The mixing time of a finite, reversible, irreducible Markov chain is
\begin{equation}
\label{eq:SobolevMixingTime}
\tau(\eps) = O\left(\frac{1}{\rho} \log \log \frac{1}{\pi_*} + \frac{1}{\Delta}{\log \frac{d}{\eps}}\right)
\end{equation}
where $\rho$ is the Sobolev constant, $\pi_*$ is the smallest value of the stationary distribution, $\Delta$ is the gap and $d$ is the size of the state space.
\end{lemma}
Further, the comparison theorem (\thmref{Comparison}) works just the same to give
\begin{equation*}
\rho \ge \hat{\rho}/A.
\end{equation*}
We will need one more result, due to Diaconis and Saloff-Coste:
\begin{lemma}[\cite{DiaconisSaloff-Coste96}, Lemma 3.2]
\label{lem:ProductChain}
Let $P_i$, $i = 1, \ldots, d$, be Markov chains with gaps $\Delta_i$ and Sobolev constants $\rho_i$.  Now construct the product chain $P$.  This chain has state space equal to the product of the spaces for the chains $P_i$ and at each step one of the chains is chosen at random and run for one step.  Then $P$ has spectral gap given by:
\begin{equation*}
\Delta = \frac{1}{d} \min_i \Delta_i
\end{equation*}
and Sobolev constant:
\begin{equation*}
\rho = \frac{1}{d} \min_i \rho_i.
\end{equation*}
\end{lemma}

\subsection{Convergence Proof}

We now prove the Markov chain convergence results to show that the $\gamma(p, p)$ terms converge quickly.  We have already shown that the $\gamma(p_1, p_2)$ terms with $p_1 \ne p_2$ converge quickly and that there is no mixing between these terms and the $\gamma(p, p)$ terms.  Therefore, in this section, we remove such terms from $\hat{G}$.

We want to prove the Markov chain with transition matrix (\eq{GeneralTransitionMatrix})
\begin{equation*}
P = \frac{1}{n(n-1)} \sum_{i \ne j} \hat{G}^{(ij)}
\end{equation*}
converges quickly.  Firstly, we know from \secref{MomentsGeneral} that $P$ has two eigenvectors with eigenvalue $1$.  The first is the identity state ($\sigma_0 \otimes \sigma_0$) and the second is the uniform sum of all non-identity terms ($\frac{1}{4^n-1}\sum_{p \ne 0} \sigma_p \otimes \sigma_p$).  From now on, we remove the identity state.  This makes the chain irreducible.  Since we know it converges, it must be aperiodic also so the chain is ergodic and all other eigenvalues are strictly between $1$ and $-1$.

We show here that the gap of this chain, up to constants, does not depend on the choice of 2-copy gapped gate set.  In the second half of the paper we find a tight bound on the gap for the $U(4)$ case which consequently gives a tight bound on the gap for all universal sets.
 
Since the stationary distribution is uniform, the chain is reversible
if and only if $P$ is a symmetric matrix.  A sufficient condition for
$P$ to be symmetric is for $\hat{G}^{(ij)}$ to be symmetric.  We saw
in \thmref{SymmetryG} that for the $U(4)$ gate set case
$\hat{G}^{(ij)}$ is symmetric.  In fact, the proof works identically
to show that $\hat{G}^{(ij)}$ is symmetric for any gate set, provided
the set is invariant under Hermitian conjugation.  However, 2-copy gapped gate
sets do not necessarily have this property so the Markov chain is not
necessarily reversible.  We will find equal bounds (up to constants)
for the gaps of both $P$ (if $\hat{G}$ is symmetric) and $P P^*$ (if
$\hat{G}$ is not symmetric) below: 
\begin{theorem}
\label{thm:GapGeneralUniversal}
Let $\mu$ be any 2-copy gapped distribution of gates.  If $\mu$ is
invariant under Hermitian conjugation then let $\Delta_P$ be the
eigenvalue gap of the resulting Markov chain matrix $P$.  Then
\begin{equation}
\Delta_P = \Omega(\Delta_{U(4)})
\end{equation}
where $\Delta_{U(4)}$ is the eigenvalue gap of the $U(4)$ chain.  If
$\mu$ is not invariant under Hermitian conjugation then let $\Delta_{P
  P^*}$ be the eigenvalue gap of the resulting Markov chain matrix $P
P^*$.  Then 
\begin{equation}
\Delta_{P P^*} = \Omega(\Delta_{U(4)}).
\end{equation}
\end{theorem}
\begin{proof}
We will use the comparison method with flows (\thmref{ComparisonFlows}).  Firstly consider the case where $\mu$ is closed under Hermitian conjugation i.e.~$\hat{G}$ is symmetric.

We will compare $P$ to the $U(4)$ chain, which we call $P_{U(4)}$.  Recall that this chain chooses a pair at random and does nothing if the pair is $00$ and chooses a random state from $\{0,1,2,3\}^2 \backslash \{00\}$ otherwise.

To apply \thmref{ComparisonFlows}, we need to construct the flows between transitions in $P_{U(4)}$.  We will choose paths such that only one pair is modified throughout.  For example (with $n=4$), the transition $1000 \rightarrow 2000$ is allowed in $P_{U(4)}$.  To construct a path in $P$, we need to find allowed transitions between these two paths in $P$.  $\hat{G}$ may not include the transition $10 \rightarrow 20$ directly, however, $\hat{G}$ is irreducible on this subspace of just two pairs.  This means that a path exists and can be of maximum length $14$ if it has to cycle through all intermediate states (in fact, since $\hat{G}$ is symmetric the maximum path length is $8$; all that is important here is that it is constant).  For example, the transitions $10 \rightarrow 11 \rightarrow 20$ might be allowed.  Then we could choose the full path to be $1000 \rightarrow 1100 \rightarrow 2000$.  In this case we have chosen the path to involve transitions pairing sites 1 and 2.  However, we could equally well have chosen any pairing; we could pair the first site with any of the others.  We can choose 3 paths in this way.  For this example, the flow we want to choose will be all 3 of these paths equally weighted.  We now use this idea to construct flows between all transitions in $P_{U(4)}$ to prove the result.

Let $x \ne y \in \Omega$ and let $d(x, y)$ be the Hamming distance between the states ($d(x, y)$ gives the number of places at which $x$ and $y$ differ).  There are two cases where $P_{U(4)}(x, y) \ne 0$:
\begin{enumerate}
\item{$d(x, y) = 2$.  Here we must choose a unique pairing, specified by the two sites that differ.  Make all transitions in $P$ using this pair giving just one path.}
\item{$d(x, y) = 1$.  For this case, choose all possible pairings of the changing site that give allowed transitions in $P_{U(4)}$.  For each pairing, construct a path in $P$ modifying only this pair.  If the differing site is initially non-zero then there are $n-1$ such pairings; if the differing site is initially zero then there are $n-z(x)$ pairings where $z(x)$ is the number of zeroes in the state $x$.}
\end{enumerate}
All the above paths are of constant length since we have to (at most) cycle through all states of a pair.  We must now choose the weighting $f(\gamma_{xy})$ for each path such that
\be
\sum_{\mathcal{P}_{xy}} f(\gamma_{xy}) = P_{U(4)}(x, y)
\ee
where $\mathcal{P}_{xy}$ is the set of all paths from $x$ to $y$ constructed above.  We choose the weighting of each path to be uniform.  We just need to calculate the number of paths in $\mathcal{P}_{xy}$ to find $f$:
\begin{enumerate}
\item{$d(x, y) = 2$.  There is just one path so $f(\gamma_{xy}) = P_{U(4)}(x, y) = \Theta(1/n^2)$.}
\item{$d(x, y) = 1$.  If the differing site is initially non-zero then $P_{U(4)}(x, y) = \Theta(1/n)$ and there are $n-1$ paths so $f(\gamma_{xy}) = \frac{P_{U(4)}(x,y)}{n-1} = \Theta(1/n^2)$.  If the differing site is initially zero then $P_{U(4)}(x, y) = \Theta\left(\frac{n-z(x)}{n^2}\right)$ and there are $n-z(x)$ paths so $f(\gamma_{xy}) = \frac{P_{U(4)}(x, y)}{n-z(x)} = \Theta(1/n^2)$.}
\end{enumerate}
So for all paths, $f = \Theta(1/n^2)$.    We now just need to know how many times each edge $(a, b)$ in $P$ is used to calculate $A$:
\begin{equation}
A = \max_{a \ne b, P(a, b) \ne 0} A(a, b)
\end{equation}
where
\begin{equation}
A(a, b) = \frac{1}{P(a, b)} \sum_{x \ne y, \gamma_{xy} \in \mathcal{P}_{xy} : (a, b) \in \gamma_{xy}} f(\gamma_{xy}).
\end{equation}
We have cancelled the factors of $\pi(x)$ because the stationary distribution is uniform.  We have also ignored the lengths of the paths since they are all constant.  

To evaluate $A(a, b)$, we need to know how many paths pass through each edge $(a, b)$.  We again consider the two possibilities separately:
%.  For each possible pairing in $P$ that can give the transition from $a$ to $b$, there is a constant number of $x, y$ pairs with paths through it, since at least one of these sites is changing.  For any $x, y$ there is at most one path $\gamma_{xy}$ using each transition pair.  We can now calculate $A$ (\eq{ComparisonFlowsA}).  We need to evaluate:
%Now we evaluate $A(a, b)$ for all pairs with $P(a, b) \ne 0$:
\begin{enumerate}
\item{$d(a, b) = 2$.  Suppose $a$ and $b$ differ at sites $i$ and $j$.  Firstly, we need to count how many transitions from $x$ to $y$ in $P_{U(4)}$ could use this edge, and then how many paths for each transition actually use the edge. 

To find which $x$ and $y$ could use the edge, note that $x$ and $y$ must differ at sites $i$, $j$ or both.  Furthermore, the values at the sites other than $i$ and $j$ must be the same as for $a$ (and therefore $b$).  There is a constant number of $x, y$ pairs that satisfy this condition.  Now, for each $x, y$ pair satisfying this, paths that use this edge must use the pairing $i, j$ for all transitions.  Since in the paths we have chosen above there is a unique path from $x$ to $y$ for each pairing, there is at most one path for each $x, y$ pair that uses edge $a, b$.

For $d(a, b) = 2$, $P(a, b) = \Theta(1/n^2)$ so $A(a, b)$ is a constant for this case.}
\item{$d(a, b) = 1$.  Let there be $r$ pairings that give allowed transitions in $P$ between $a$ and $b$.  As above, each pairing gives a constant number of paths.  So the numerator is $\Theta(r/n^2)$.  Further, $P(a, b) = \Theta(r/n^2)$.  So again $A(a, b)$ is constant.}
\end{enumerate}
Combining, $A$ is a constant so the result is proven for the case $\hat{G}$ is symmetric.

We now turn to the irreversible case.  We now need to bound the gap of $P P^* = P P^T$.  This chain selects two (possibly overlapping) pairs at random and applies $\hat{G}$ to one of them and $\hat{G}^T$ to the other.  We can use the above exactly by choosing $\hat{G}$ to perform the transitions above and $\hat{G}^T$ to just loop the states back to themselves.  By aperiodicity (the greatest common divisor of loop lengths is $1$), we can always find constant length paths that do this.
\end{proof}

Now we need to know the gap of the $U(4)$ chain.  We can, by a simple application of the comparison theorem, show it is $\Omega(1/n^2)$.  However, in the second half of this paper we show it is $\Theta(1/n)$.  This gives us (using \thmref{MixingTimeGap}):
\begin{corollary}
\label{cor:SecondMomentsMixing}
The Markov chain $P$ has mixing time $O(n(n+ \log 1/\eps))$ and 2-norm mixing time $O(n\log 1/\eps)$.
\end{corollary}
We conjecture that the mixing time (as well as \lemref{CoefficientsDecay2General}) can be tightened to $\Theta(n\log\frac{n}{\eps})$, which is asymptotically the same as for the $U(4)$ case:
\begin{conjecture}
\label{conj:Mixing}
The second moments for the case of general 2-copy gapped distributions
have 1-norm mixing time $\Theta(n\log\frac{n}{\eps})$. 
\end{conjecture}
It seems likely that an extension of our techniques in \secref{U4Convergence} could be used to prove this.

Combining the convergence results we have proved our general result \lemref{MainMixing}:
\begin{proof}[of \lemref{MainMixing}]
Combining \corref{SecondMomentsMixing} (for the $\gamma(p, p)$ terms) and \lemref{CoefficientsDecay2General} (for the $\gamma(p_1, p_2)$, $p_1 \ne p_2$ terms) proves the result.
\end{proof}

We have now shown that the first and second moments of random circuits converge quickly.  For the remainder of the paper we prove the tight bound for the gap and mixing time of the $U(4)$ case and show how mixing time bounds relate to the closeness of the 2-design to an exact design.  Only for the $U(4)$ case is the matrix $\hat{G}$ a projector so in this sense the $U(4)$ random circuit is the most fundamental.  While we expect the above mixing time bound is not tight, we can prove a tight mixing time result for the $U(4)$ case.  However, using our definition of an approximate $k$-design, the gap rather than the mixing time governs the degree of approximation.

\section{Tight Analysis for the \texorpdfstring{$U(4)$}{U(4)} Case}
\label{sec:U4Convergence}

We have already found tight bounds for the first moments in \lemref{CoefficientDecayGeneral}: just set $\Delta = 1$.

\subsection{Second Moments Convergence}

We need to prove a result analogous to \lemref{CoefficientsDecay2General} for the terms $\sigma_{p_1} \otimes \sigma_{p_2}$ where $p_1 \ne p_2$.  We already have a tight bound for the 2-norm decay, by setting $\Delta = 1$ into \lemref{CoefficientsDecay2General}.  We tighten the 1-norm bound:
\begin{lemma}
\label{lem:CoefficientsDecay2}
After $O(n \log \frac{n}{\eps})$ steps
\begin{equation}
\label{eq:TermsDecay}
\sum_{p_1 \ne p_2} \Expect_W | \gamma_W(p_1, p_2) | \le \eps
\end{equation}
\end{lemma}
\begin{proof}
We will split the random circuits up into classes depending on how many qubits have been hit.  Let $H$ be the random variable giving the number of different qubits that have been hit.  We can work out the distribution of $H$ and bound the sum of $| \gamma_W(p_1, p_2) |$ for each outcome.

Firstly we have, after $t$ steps,
\begin{equation*}
\Pr(H \le h) \le {n \choose h} \left(\frac{h(h-1)}{n(n-1)}\right)^t \le {n \choose h} (h/n)^t.
\end{equation*}
Now, for each qubit hit, each coefficient which has $p_1$ and $p_2$ differing in this place is set to zero.  So after $h$ have been hit, there are only (at most) $16^{(n-h)}$ terms in the sum in \eq{TermsDecay}.  As before, the state is a physical state, $\tr \rho^2 \le 1$ so $\sum_{p_1 p_2} \gamma^2(p_1, p_2) \le 1$ so $\sum_{p_1 p_2} | \gamma(p_1, p_2) | \le \sqrt{N}$ if there are at most $N$ non-zero terms in the sum.  Therefore we have, after $t$ steps,
\begin{align*}
\sum_{p_1 \ne p_2} \Expect_W | \gamma_W(p_1, p_2) | &\le \sum_{h=1}^{n-1} \Pr(H = h) 16^{(n-h)/2} \\
&\le \sum_{h=1}^{n-1} \Pr(H \le h) 4^{(n-h)} \\
&\le \sum_{h=1}^{n-1} {n \choose h} (h/n)^t 4^{(n-h)} \\
&= \sum_{h=1}^{n-1} {n \choose h} (1-h/n)^t 4^{h} \qquad h \rightarrow n-h  \\
&\le \sum_{h=1}^{n-1} {n \choose h} \exp(-ht/n) 4^{h}.
\end{align*}
Now, let $t = n \ln \frac{n}{\eps}$:
\begin{align*}
\sum_{p_1 \ne p_2} \Expect_W | \gamma_W(p_1, p_2) | &\le \sum_{h=1}^{n-1} {n \choose h} \left(\frac{4 \eps}{n} \right)^h \\
&= \left( 1 + \frac{4\eps}{n} \right)^n -1 -\left( \frac{4\eps}{n} \right)^n = O(\eps)
\end{align*}
where the last line follows from the binomial theorem.
\end{proof}
This, combined with the mixing time result we prove below, completes the proof that the second moments of the random circuit converge in time $O(n \log \frac{n}{\eps})$.

\subsection{Markov Chain of Coefficients}

The Markov chain acting on the coefficients is reducible because the state $\{ 0 \}^n$ is isolated.  However, if we remove it then the chain becomes irreducible.  The presence of self loops implies aperiodicity therefore the chain is ergodic.  We have already seen that the chain converges to the Haar uniform distribution (in \secref{RandomCircuits}) therefore the stationary state is the uniform state $\pi(x) = 1/(4^n-1)$.  Further, since the chain is symmetric and has uniform stationary distribution, the chain satisfies detailed balance (\eq{DetailedBalance}) so is reversible.  We now turn to obtaining bounds on the mixing time of this chain.

We want to show that the full chain converges to stationarity in time $\Theta(n \log \frac{n}{\eps})$.  This implies (see later) that the gap is $\Theta(1/n)$.  To prove this, we will construct another chain called the zero chain.  This is the chain that counts the number of zeroes in the state.  Since it is the zeroes that slow down the mixing, this chain will accurately describe the mixing time of the full chain.
\begin{lemma}
\label{lem:ZeroChainTransitionMatrix}
The zero chain has transition matrix P on state space (we count non-zero positions) $\Omega = \{1,2, \ldots, n\}$.
\begin{equation}
P(x, y) =
\begin{cases}
  1 - \frac{2x(3n-2x-1)}{5n(n-1)}  & y = x \\
  \frac{2x(x-1)}{5n(n-1)}  & y = x-1 \\
  \frac{6x(n-x)}{5n(n-1)}  & y = x+1 \\
  0 & \rm{otherwise}
\end{cases}
\end{equation}
for $1 \le x, y \le n$.
\end{lemma}
\begin{proof}
Suppose there are $n-x$ zeroes (so there are $x$ non-zeroes).  Then the only way the number of zeroes can decrease (i.e.~for $x$ to increase) is if a non-zero item is paired with a zero item and one of the $9$ (out of $15$) new states is chosen with no zeroes.  The probability of choosing such a pair is $\frac{2x(n-x)}{n(n-1)}$ so the overall probability is $\frac{9}{15} \frac{2x(n-x)}{n(n-1)}$.

The number of zeroes can increase only if a pair of non-zero items is chosen and one of the $6$ states is chosen with one zero.  The probability of this occurring is $\frac{6}{15} \frac{x(x-1)}{n(n-1)}$.

The probability of the number of zeroes remaining unchanged is simply calculated by requiring the probabilities to sum to $1$.
\end{proof}
We see that the zero chain is a one-dimensional random walk on the line.  It is a lazy random walk because the probability of moving at each step is $<1$.  However, as the number of zeroes decreases, the probability of moving increases monotonically:
\begin{equation}
1-P(x,x) = \frac{2x(3n-2x-1)}{5n(n-1)} \ge 2x/5n <1.
\end{equation}

\begin{lemma}
\label{lem:ZeroStatDistrib}
The stationary distribution of the zero chain is
\begin{equation}
\label{eq:ZeroStatDistrib}
\pi_0(x) = \frac{3^x {n \choose x}}{4^n-1}.
\end{equation}
\end{lemma}
\begin{proof}
This can be proven by multiplying the transition matrix in \lemref{ZeroChainTransitionMatrix} by the state \eq{ZeroStatDistrib}.  Alternatively, it can be proven by counting the number of states with $n-x$ zeroes.  There are ${n \choose x}$ ways of choosing which sites to make non-zero and each non-zero site can be one of three possibilities: 1, 2 or 3.  The total number of states is $4^n-1$, which gives the result.
\end{proof}

Below we will prove the following theorem:
\begin{theorem}
\label{thm:ZeroChainMixing}
The zero chain mixes in time $\Theta(n \log \frac{n}{\eps})$.
\end{theorem}
The 2-norm mixing time follows easily:
\begin{theorem}
\label{thm:ZeroChainMixing2Norm}
The zero chain has 2-norm mixing time $O(n \log 1/\eps)$.
\end{theorem}
\begin{proof}
We use a lower bound on the 1-norm mixing time to show that the gap of the zero chain is $\Omega(1/n)$ and then use the 2-norm mixing bound \eq{2normMixing}.  In \cite{MontenegroTetali06}, Theorem 4.9, they prove the lower bound:
\be
\tau_1(\eps) \ge \frac{1-\Delta}{\Delta} \ln \frac{1}{2\eps}
\ee
where $\Delta$ is the eigenvalue gap.  In \thmref{ZeroChainMixing}, we showed $\tau_1(\eps) \le Cn \ln \frac{n}{\eps}$ for some constant $C$.  Combining,
\be
\frac{1-\Delta}{\Delta} \ln \frac{1}{2\eps} \le Cn \ln \frac{n}{\eps}
\ee
for all $\eps>0$.  Divide by $\ln 1/\eps$ and take the limit $\eps \rightarrow 0$ to find
\be
\frac{1-\Delta}{\Delta} \le Cn
\ee
which implies the gap is $\Omega(1/n)$.  The 2-norm bound now follows from \eq{2normMixing}.
\end{proof}
Before proving \thmref{ZeroChainMixing}, we will show how the mixing time of the full chain follows from this.

\begin{corollary}
\label{cor:FullChainMixing}
The full chain mixes in time $\Theta(n \log \frac{n}{\eps})$.
\end{corollary}
\begin{proof}
Once the zero chain has approximately mixed, the distribution of zeroes is almost correct.  We need to prove that the distribution of non-zeroes is correct after $O(n \log \frac{n}{\eps})$ steps too.

Once each site of the full chain has been hit, meaning it is chosen and paired with another site so not both equal zero, the chain has mixed.  This is because, after each site has been hit, the probability distribution over the states is uniform.  When the zero chain has approximately mixed, a constant fraction of sites are zero so the probability of hitting a site at each step is $\Theta(1/n)$.  By the coupon collector argument, each site will have been hit with probability at least $1-\eps$ in time time $O(n \log \frac{n}{\eps})$.  Once the zero chain has mixed to $\eps'$, we can run the full chain this extra number of steps to ensure each site has been hit with high probability.  Since the mixing of the zero chain only increases with time, the distance to stationarity of the full chain is now $1-\eps-\eps'$.  We make this formal below.

After $t_0 = O(n \log \frac{n}{\eps'})$ steps, the number of zeroes is $\eps'$-close to the stationary distribution $\pi_0$ by \thmref{ZeroChainMixing} and only gets closer with more steps since the distance to stationarity decreases monotonically.  The stationary distribution \eq{ZeroStatDistrib} is approximately a Gaussian peaked at $3n/4$ with $O(n)$ variance.  This means that, with high probability, the number of non-zeroes is close to $3n/4$.  We will in fact only need that there is at least a constant fraction of non-zeroes; with probability at least $1-\eps'-\exp(-\Omega(n))$ there will be at least $n/2$.

To prove the mixing time, we run the chain for time $t_0$ so the zero chain mixes to $\eps'$.  Then run for $t_1$ additional steps.  Let $H_{i, t}$ be the event that site $i$ is hit at step $t$.  Let $H_i = \cup_{t = t_0+1}^{t_0 +t_1} H_{i, t}$ and $H = \cap_{i = 1}^n H_i$.  We want to show $\Pr(H)$ is close to 1, or, in other words, that all sites are hit with high probability.  Further let $X_t$ be the random variable giving the number of non-zeroes at step $t$.

If at step $t-1$ site $i$ is non-zero then the event $H_{i, t}$ occurs if the qubit is chosen, which occurs with probability $2/n$.  If, however, it was zero then it must be paired with a non-zero thing for $H_{i, t}$ to hold.  Conditioned on any history with $X_{t-1} \ge n/2$, this probability is $\ge 1/n$.  In particular, we can condition on not having previously hit $i$ and the bound does not change.  Combining we have
\begin{equation*}
\Pr\left(H_{i, t}^c \bigg| \left[ X_{t-1} \ge n/2 \right] \bigcap \left( \bigcap_{t' = t_0+1}^{t-1} H_{i, t'}^c \right)\right) \le 1 - 1/n.
\end{equation*}
Then, after $t_1$ extra steps,
\begin{equation*}
\Pr\left(H_i^c \bigg| \bigcap_{t = t_0}^{t_0+t_1-1} \left[ X_{t} \ge n/2 \right] \right) \le (1 - 1/n)^{t_1}
\end{equation*}
which, using the union bound, gives
\begin{equation*}
\Pr\left(H^c \bigg| \bigcap_{t = t_0}^{t_0+t_1-1} \left[ X_{t} \ge n/2 \right] \right) \le n(1 - 1/n)^{t_1}.
\end{equation*}
Now, since the zero chain has mixed to $\eps'$,
\begin{equation*}
\Pr\left(\bigcap_{t = t_0}^{t_0+t_1-1} \left[ X_{t} \ge n/2 \right]\right) \le t_1 \sum_{x=n/2}^{n-1} \pi_0(x) + \eps' \le t_1 \exp(-O(n)) + \eps'
\end{equation*}
so
\begin{equation*}
\Pr(H^c) \le n(1 - 1/n)^{t_1} + t_1 \exp(-O(n)) + \eps'.
\end{equation*}
Now, choose $t_1 = n \ln \frac{2n}{\eps}$ so that $\Pr(H^c) \le \delta$ where $\delta = \eps + t_1 \exp(-O(n))$.  Choose $\eps = 1/n$ so that $\delta$ is $1/\poly(n)$.  Now, using the bound on $\Pr(H^c)$, we can write the state $v$ after $t_1 = O(n \log n)$ steps as
\begin{equation*}
v = (1-\delta) \pi + \delta \pi'
\end{equation*}
where $\pi$ is the stationary distribution and $\pi'$ is any other distribution.  Using this,
\begin{equation*}
|| v- \pi || \le \delta.
\end{equation*}
We now apply \lemref{DistanceToStat} to show that after $O(n \log \frac{n}{\eps})$ steps the distance to stationarity of the full chain is $\eps$.
\end{proof}

\subsection{Proof of Theorem 5.4} %\thmref{ZeroChainMixing}}

We will now proceed to prove \thmref{ZeroChainMixing}.  We present an outline of the proof here; the details are in \secref{ZeroChainMixingProof}.

Firstly, note that by the coupon collector argument, the lower bound on the time is $\Omega (n \log n)$.  We need to prove an upper bound equal to this.  Intuition says that the mixing time should take time $O(n \log n)$ because the walk has to move a distance $\Theta(n)$ and the waiting time at each step is proportional to $n, n/2, n/3, \ldots$ which sums to $O(n \log n)$, provided each site is not hit too often.   We will show that this intuition is correct using Chernoff bound and log-Sobolev (see later) arguments.

We will first work out concentration results of the position after some number of \emph{accelerated} steps.  The zero chain has some probability of staying still at each step. The accelerated chain is the zero chain conditioned on moving at each step.  We define the accelerated chain by its transition matrix:
\begin{definition}
The transition matrix for the accelerated chain is
\begin{equation}
P_a(x, y) =
\begin{cases}
  0  & y = x \\
  \frac{x-1}{3n-2x-1} & y = x-1 \\
  \frac{3(n-x)}{3n-2x-1} & y = x+1 \\
  0 & \rm{otherwise}.
\end{cases}
\end{equation}
\end{definition}
We use the accelerated chain in the proof to firstly prove the accelerated chain mixes quickly, then to bound the waiting time at each step to obtain a mixing time bound for the zero chain.
  
To prove the mixing time bound, we will split the walk up into three phases.  We will split the state space into three (slightly overlapping) parts and the phase can begin at any point within that space. So each phase has a state space $\Omega_i \subset [1, n]$, an entry space $E_i \subset \Omega_i$ and an exit condition $T_i$.  We say that a phase completes successfully if the exit condition is satisfied in time $O(n \log n)$ for an initial state within the entry space.  When the exit condition is satisfied, the walk moves onto the next phase.

The phases are:
\begin{enumerate}
\item{$\Omega_1 = [1, n^\delta]$ for some constant $\delta$ with $0 < \delta < 1/2$.  $E_1 = \Omega_1$ (i.e.~it can start anywhere) and $T_1$ is satisfied when the walk reaches $n^\delta$.  For this part, the probability of moving backwards (gaining zeroes) is $O(n^{\delta -1})$ so the walk progresses forwards at each step with high probability.  This is proven in \lemref{Phase1MovesRight}.  We show that the waiting time is $O(n \log n)$ in \lemref{Phase1WaitingTime}.}
\item{$\Omega_2 = [n^\delta/2, \theta n]$ for some constant $\theta$ with $0 < \theta < 3/4$.  $E_2 = [n^\delta, \theta n]$ and $T_2$ is satisfied when the walk reaches $\theta n$.  Here the walk can move both ways with constant probability but there is a $\Omega(1)$ forward bias.  Here we use a monotonicity argument: the probability of moving forward at each step is
\begin{align*}
p(x) &= \frac{3(n-x)}{3n-2x-1} \\
& \ge \frac{3(n-x)}{3n-2x} \\
& \ge \frac{3(1-\theta)}{3-2\theta}.
\end{align*}
If we model this random walk as a walk with constant bias equal to $\frac{3(1-\theta)}{3-2\theta}$ we will find an upper bound on the mixing time since mixing time increases monotonically with decreasing bias.  Further, the waiting time at $x=a$ stochastically dominates the waiting time at $x=b$ for $b \ge a$.  The true bias decreases with position so the walk with constant bias spends more time at the early steps.  Thus the position of this simplified walk is stochastically dominated by the position of the real walk while the waiting time stochastically dominates the waiting time of the real walk.}
\item{$\Omega_3 = [\frac{\theta}{2} n, n]$ and $E_3 = [\theta n, n]$.  $T_3$ is satisfied when this restricted part of the chain has mixed to distance $\eps$.  Here the bias decreases to zero as the walk approaches $3n/4$ but the moving probability is a constant.  We show that this walk mixes quickly by bounding the log-Sobolev constant of the chain.}
\end{enumerate}
Showing these three phases complete successfully will give a mixing time bound for the whole chain.

We now prove in the Appendix that the phases complete successfully with probability at least $1-1/\poly(n)$:
\begin{lemma}
\label{lem:Phase1Completes}
\begin{equation*}
\Pr(\text{\rm{Phase 1 completes successfully}}) \ge 1 - n^{2\delta-1} - 2 n^{-\delta}
\end{equation*}
\end{lemma}

\begin{lemma}
\label{lem:Phase2Completes}
\begin{equation*}
\Pr(\text{\rm{Phase 2 completes successfully}}) \ge 1 - \exp\left(-\frac{2}{3} \mu \theta n\right) - \left(\frac{4}{\theta n}\right)^\frac{3}{2\mu} - \frac{2 \exp\left(\frac{-\mu n^\delta}{4}\right)}{1-\exp(-\mu/2)} - \left(q/p \right)^{n^\delta/2}
\end{equation*}
where $\mu = \frac{6(1-\theta)}{3-2\theta} -1$.
\end{lemma}

\begin{lemma}
\label{lem:Phase3Completes}
\begin{equation*}
\Pr(\text{\rm{Phase 3 completes successfully}})  \ge 1 - \left( \frac{\theta}{3(2-\theta)} \right)^{\theta n/2}
\end{equation*}
\end{lemma}

We can now finally combine to prove our result:
\begin{proof}[of \thmref{ZeroChainMixing}]
The stationary distribution has exponentially small weight in the tail with lots of zeroes.  We show that, provided the number of zeroes is within phase 3, the walk mixes in time $O(n \log \frac{n}{\eps})$.  We also show that if the number of zeroes is initially within phase 1 or 2, after $O(n \log n)$ steps the walk is in phase 3 with high probability.  We can work out the distance to the stationary distribution as follows.

Let $p_f$ be the probability of failure.  This is the sum of the error probabilities in Lemmas \ref{lem:Phase1Completes}, \ref{lem:Phase2Completes} and \ref{lem:Phase3Completes}.  The key point is that $p_f = 1/\poly(n)$.  Then after $O(n \log \frac{n}{\eps})$ steps (the sum of the number of steps in the 3 phases), the state is equal to $(1-p_f)v_3 + p_f v'$ where $v_3$ is the state in the phase 3 space and $v'$ is any other distribution, which occurs if any one of the phases fails.  Since the distance to stationarity in phase 3 is $\eps$, $|| v_3 - \pi_3 || \le \eps$, where $\pi_3$ is the stationary distribution on the state space of phase 3.  In \lemref{MixingPhase3} we show that $\pi_3(x) = \pi(x)/(1-w)$ where $w = \sum_{x=1}^{\theta n / 2 -1} \pi(x)$.  Since $\pi(x)$ is exponentially small in this range, $w$ is exponentially small in $n$.  Now use the triangle inequality to find
\begin{equation}
||v_3 - \pi|| \le ||v_3 - \pi_3|| + ||\pi_3 - \pi||.
\end{equation}
Since the chain in phase 3 has mixed to $\eps$, the first term is $\le \eps$.  We can evaluate $|| \pi_3 - \pi||$:
\begin{align*}
|| \pi_3 - \pi || &= \frac{1}{2} \sum_{x=1}^n || \pi_3(x) - \pi(x) || \\
&= \frac{1}{2} \left( \sum_{x=1}^{\theta n/2 - 1} \pi(x) + \sum_{x = \theta n /2}^n(\pi(x)/(1-w) - \pi(x)) \right) \\
&= \frac{1}{2} \left( w + 1 - (1-w) \right) = w.
\end{align*}
So now,
\begin{align*}
||(1-p_f)v_3 + p_f v' - \pi|| &= ||(1-p_f)(v_3 - \pi) + p_f(v' - \pi)|| \\
&\le (1-p_f) || v_3 - \pi|| + p_f ||v' - \pi|| \\
&\le (1-p_f)(\eps + w) + p_f \\
&\le \delta
\end{align*}
where $\delta = \eps + w + p_f$.  We are free to choose $\eps$: choose it to be $1/n$ so that $\delta$ is $1/\poly(n)$.  So now the running time to get a distance $\delta$ is $t = O(n \log n)$.  We then apply \lemref{DistanceToStat} to obtain the result.

This concludes the proof of \thmref{ZeroChainMixing} so \corref{FullChainMixing} is proved.
\end{proof}
We have now proven \lemref{MainMixing} and consequently \corref{MainMixing}.  We now show how \thmref{Main2Design} follows.

\section{Main Result}
\label{sec:MainResult}

We will now show how the mixing time results imply that we have an approximate 2-design.

\begin{proof}[Proof of \thmref{Main2Design}:]
We will go via the 2-norm since this gives a tight bound when working
with the Pauli operators.  The supremum can be taken over just
physical states $\rho$ \cite{Watrous05}.  We write $\rho$ in the Pauli
basis as usual (as \eq{PauliBasisGeneral}).
\begin{align*}
\vectornorm{\cG_{W} - \cG_H}_{\diamond}^2 &= \sup_\rho \vectornorm{(\cG_{W} \otimes I) (\rho) - (\cG_H \otimes I) (\rho)}_1^2 \\
&\le 2^{4n} \sup_\rho \vectornorm{(\cG_{W} \otimes I) (\rho) - (\cG_H \otimes I) (\rho)}_2^2 \\
&= \sup_{\rho} \bigg|\bigg| \sum_{p_1, p_2, p_3, p_4 \atop p_1 p_2 \ne 00} \gamma_0(p_1, p_2, p_3, p_4) ( \cG_{W}(\sigma_{p_1} \otimes \sigma_{p_2}) \otimes \sigma_{p_3} \otimes \sigma_{p_4}  \\
& \phantom{= \sum_{\rho} \bigg|\bigg|} -  \cG_H(\sigma_{p_1} \otimes \sigma_{p_2}) \otimes \sigma_{p_3} \otimes \sigma_{p_4} ) \bigg|\bigg|_2^2
\end{align*}
Now, write (for $p_1 p_2 \ne 00$) $\cG_{W}(\frac{1}{2^n}\sigma_{p_1} \otimes \sigma_{p_2}) = \frac{1}{2^n} \sum_{q_1, q_2 \atop q_1 q_2 \ne 00} g_t(q_1, q_2; p_1, p_2) \sigma_{q_1} \otimes \sigma_{q_2}$.  We get
\begin{align*}
& \sup_{\rho} \bigg|\bigg| \sum_{p_1, p_2, p_3, p_4, q_1, q_2 \atop p_1 p_2 \ne 00, q_1 q_2 \ne 00} \gamma_0(p_1, p_2, p_3, p_4) \left( g_t(q_1, q_2; p_1, p_2) - \frac{\delta_{q_1 q_2} \delta_{p_1 p_2}}{2^n(2^n+1)} \right) \\
& \phantom{\sup_{\rho} \bigg|\bigg|} \sigma_{q_1} \otimes \sigma_{q_2} \otimes \sigma_{p_3} \otimes \sigma_{p_4} \bigg|\bigg|_2^2 \\
&= 2^{4n} \sup_{\rho} \sum_{p_1, p_2, p_3, p_4, q_1, q_2 \atop p_1 p_2 \ne 00, q_1 q_2 \ne 00} \gamma_0^2(p_1, p_2, p_3, p_4) \left( g_t(q_1, q_2; p_1, p_2) - \frac{\delta_{q_1 q_2} \delta_{p_1 p_2}}{2^n(2^n+1)} \right)^2  \\
&\le 2^{4n} \sup_{\rho} \sum_{p_1, p_2, p_3, p_4 \atop p_1 p_2 \ne 00} \gamma_0^2(p_1, p_2, p_3, p_4) \eps^2  \\
&\le 2^{4n} \eps^2
\end{align*}
where the first equality comes from the orthogonality of the Pauli operators under the Hilbert-Schmidt inner product and the last inequality comes from the fact that $\rho$ is a physical state so has $\tr \rho^2 \le 1$.  This proves the result for the diamond norm, \defref{ApproxUnitaryDesign}.  For the distance measure defined in \defref{DankertApprox}, the argument in \cite{DCEL06} can be used together with the 1-norm bound to prove the result.
\end{proof}
It is unfortunate that there is still a dimension factor remaining in
the above proof.  To get a distance $\eps$ we have to run the random
circuit for $O(n(n+\log 1/\eps))$ steps.  However, closeness in the
diamond-norm may be too stringent a requirement.  After
$O(n(n+\log 1/\eps))$ steps, the random circuit gives a 2-design in the measure used
by Dankert et al.~(see \cite{DCEL06} and \defref{DankertApprox}).
This is in contrast to the $O(n \log 1/\eps)$ steps required by the
explicit circuit construction of Dankert et al.

%Another important thing to notice is that a $k$-design on the two
%qubits is all that is required to have the $\hat{G}$ (\eq{G}) for the
%$U(4)$ gate set, rather than a fully random $U(4)$ gate.  We could,
%for example, use the Clifford group for this purpose \cite{DLT02} for
%$k=2$.  Therefore this protocol takes an exact 2-design in four
%dimensions and makes an approximate 2-design in $2^n$ dimensions.

\section{Conclusions}
\label{sec:Conclusion}

We have proved tight convergence results for the first two moments of
a random circuit.  We have used this to show that random circuits are
efficient approximate 1- and 2-unitary designs.  Our framework readily
generalises to $k$-designs for any $k$ and the next step in this
research is to prove that random circuits give approximate $k$-designs
for all $k$.

We have shown that, provided the random circuit uses gates from a
universal gate set that is also universal on $U(4)$, the circuit is
still an efficient 2-design.  We also see that the random circuit with
gates chosen uniformly from $U(4)$ is the most natural model.  We note that the gates from $U(4)$ can be replaced by
gates from any approximate 2-design on two qubits without any change to the asymptotic 
convergence properties.

One application of this work is to give an efficient method of
decoupling two quantum systems by applying a random unitary from a
2-design to one system and then discarding part of it.  This technique
is used in \cite{ADHW06} to construct a variety of encoding circuits
for tasks in quantum Shannon theory; thus, we (like \cite{DCEL06})
reduce the encoding 
complexity in \cite{ADHW06} (and related works, such as \cite{HHYW07}) to
$O(n^2)$.  Unfortunately, the decoding circuits still remain
inefficient.

An algorithmic application of random circuits was given in
\cite{HH08}, where they were used to construct a new class of
superpolynomial quantum speedups.  In that paper, random circuits of
length $O(n^3)$ were used in order to guarantee that they were
so-called ``dispersing'' circuits.  Our results immediately imply that
circuits of length $O(n^2)$ would instead suffice.   We believe that
this could be further improved with a specialised argument, since
\cite{HH08} assumed that the input to the random circuit was always a
computational basis state.

Another potential application of random circuits is to model the
evolution of black holes \cite{HaydenPreskill07}.  In
Ref.~\cite{HaydenPreskill07}, they conjecture that short random local
quantum circuits are approximately 2-designs, and thus can be used for
decoupling quantum systems (as in \cite{ADHW06}).  This, in turn, is
used to make claims about the rate at which black holes leak
information.  While our model differs from that of
Ref.~\cite{HaydenPreskill07} in that they consider nearest-neighbour
interactions and we do not, our techniques and results could be
readily extended to cover the case they consider. 

Finally, random circuits are interesting physical models in their own
right.  The original purpose of \cite{ODP06} was to answer the
physical question of how quickly entanglement grows in a system with
random two party interactions.  \lemref{MainMixing}(i) shows that
$O(n(n + \log 1/\eps))$ steps suffice (in contrast to $O(n^2(n + \log
1/\eps))$ which they prove) to give almost maximal entanglement in
such a system.

{\bf Acknowledgements.} We are grateful for funding from the Army
Research Office under grant 
W9111NF-05-1-0294, the European Commission under Marie Curie grants
ASTQIT (FP6-022194) and QAP (IST-2005-15848), and the U.K. Engineering
and Physical Science Research Council through ``QIP IRC.''   We thank
Rapha\"{e}l Clifford, Ashley Montanaro and Dan Shepherd for helpful
discussions. 

\appendix

\section{Appendix}

\subsection{Permutation Operators}

The following theorems about permutation operators will be used repeatedly.
\begin{lemma}
\label{lem:TraceCycles}
Let $C$ be a cycle of length $c$ in $S_c$.  Then
\begin{equation*}
\tr \l( C \l( A_1 \otimes A_2 \otimes \ldots \otimes A_c \r)\r) = \tr \l(A_{C(1)} A_{C^{\circ 2}(1)} A_{C^{\circ 3}(1)} \ldots A_1 \r).
\end{equation*}
\end{lemma}
\begin{proof}
We have
\begin{align*}
\tr \l( C \l( A_1 \otimes A_2 \otimes \ldots \otimes A_c \r)\r) &= \sum_{i_1, i_2, \ldots, i_c} \bra{i_1 i_2 \ldots i_c} C \l( A_1 \otimes A_2 \otimes \ldots \otimes A_c \r) \ket{i_1 i_2 \ldots i_c} \\
&= \sum_{i_1, i_2, \ldots, i_c} \bra{i_1} A_{C(1)} \ket{i_{C(1)}} \bra{i_2} A_{C(2)} \ket{i_{C(2)}} \ldots \bra{i_c} A_{C(c)} \ket{i_{C(c)}} \\
&= \sum_{i_1, i_2, \ldots, i_c} \bra{i_1} A_{C(1)} \ket{i_{C(1)}} \bra{i_{C(1)}} A_{C^{\circ 2}(1)} \ket{i_{C^{\circ 2}(1)}} \ldots \bra{i_{C^{\circ c-1}(1)}} A_{1} \ket{i_1}
\end{align*}
since $C^{\circ c}(1)=1$.  Evaluate the sum using the resolution of the identity to get the result.
\end{proof}

With this we can work out the Pauli expansion of the swap operator:
\begin{lemma}
\label{lem:Swap}
The swap operator $\swap$ on two $d$ dimensional systems can be written as
\begin{equation*}
\frac{1}{d} \sum_{p} \sigma_p \otimes \sigma_p.
\end{equation*}
where $\{ \sigma_p \}$ form a Hermitian orthogonal basis with $\tr \sigma_p^2 = d$.
\end{lemma}
\begin{proof}
Expand $\swap$ in the basis and use \lemref{TraceCycles}:
\begin{align*}
\tr \sigma_p \otimes \sigma_q \swap &= \tr \sigma_p \sigma_q \\
&=
\begin{cases}
d	&	p=q \\
0	&	{\rm otherwise}.
\end{cases}
\end{align*}
The given sum has the correct coefficients in the basis therefore $\frac{1}{d} \sum_{p} \sigma_p \otimes \sigma_p=\swap$.
\end{proof}

\subsection{Zero chain mixing time proofs}
\label{sec:ZeroChainMixingProof}

\subsubsection{Asymmetric Simple Random Walk}
\label{sec:AsymRandomWalks}

We will use some facts about asymmetric simple random walks i.e.~a random walk on a 1D line with probability $p$ of moving right at each step and probability $q=1-p$ of moving left.

The position of the walk after $k$ steps is tightly concentrated around $k(p-q)$:
\begin{lemma}
\label{lem:BiasedWalkPosition}
Let $X_k$ be the random variable giving the position of a random walk after $k$ steps starting at the origin with probability $p$ of moving right and probability $q = 1-p$ of moving left.  Let $\mu = p - q$.  Then for any $\eta > 0$,
\begin{equation*}
\Pr(X_k \ge \mu k + \eta) \le \exp\left(-\frac{\eta^2}{2 k}\right)
\end{equation*}
and
\begin{equation*}
\Pr(X_k \le \mu k - \eta) \le \exp\left(-\frac{\eta^2}{2 k}\right).
\end{equation*}
\end{lemma}
\begin{proof}
The standard Chernoff bound for $0/1$ variables $\tilde{Y}_i$ gives, with $\tilde{Y}_i$ equal to $1$ with probability p and for $Y_k = \sum_{i=1}^k \tilde{Y}_i$,
\begin{align*}
\Pr(Y_k \ge k p + \eta) &\le \exp\left(-\frac{2 \eta^2}{k} \right) \\
\Pr(Y_k \le k p - \eta) &\le \exp\left(-\frac{2 \eta^2}{k} \right).
\end{align*}
For our case, set $\tilde{Y_i} = 2\tilde{X_i} - 1$ to give the desired result.
\end{proof}
This result is for a walk with constant bias.  We will need a result for a walk with varying (but bounded from below) bias:
\begin{lemma}
\label{lem:BiasedWalkPositionWithBoundedBias}
Let $X_k$ be the random variable giving the position of a random walk after $k$ steps starting at the origin with probability $p_i \ge p$ of moving right and probability $q_i \le p$ of moving left at step $i$.  Let $\mu = p - (1-p)$.  Then for any $\eta > 0$,
\begin{equation*}
\Pr(X_k \ge \mu k + \eta) \le \exp\left(-\frac{\eta^2}{2 k}\right)
\end{equation*}
and
\begin{equation*}
\Pr(X_k \le \mu k - \eta) \le \exp\left(-\frac{\eta^2}{2 k}\right).
\end{equation*}
\end{lemma}
\begin{proof}
Let $\tilde{Y}_i$ be a random variable equal to $1$ with probability $p$ and $0$ with probability $1-p$.  Then let $\tilde{Z}_i$ be a random variable equal to $1$ with probability $p_i$ and $0$ with probability $1-p_i$.  Let $Y_k = \sum_{i=1}^k \tilde{Y}_i$ and $Z_k = \sum_{i=1}^k \tilde{Z}_i$.  Then following the standard Chernoff bound derivation (for $\lambda > 0$),
\begin{align*}
\Pr(Z_k \ge k p + \eta) &= \Pr\left(e^{\lambda Z_k} \ge e^{\lambda(kp+\eta)}\right) \\
&\le \frac{e^{\lambda(k p + \eta)}}{\Expect e^{\lambda Z_k}} \\
&\le \frac{e^{\lambda(k p + \eta)}}{\Expect e^{\lambda Y_k}} \\
&\le \exp\left(-\frac{2 \eta^2}{k} \right).
\end{align*}
We can then, as above, set $\tilde{Z_i} = 2\tilde{X_i} - 1$.  The calculation is similar for the bound on $\Pr(X_k \le \mu k - \eta)$.
\end{proof}

From \lemref{BiasedWalkPosition} we can prove a result about how often each site is visited.  If the walk runs for $t$ steps the walk is at position $t \mu$ with high probability so we might expect from symmetry that each site will have been visited about $1/\mu$ times.  Below is a weaker concentration result of this form but is strong enough for our purposes.  It says that the amount of time spent $\le x$ is about $x/\mu$.
\begin{lemma}
\label{lem:NumberOfHits}
For $\gamma > 2$ and integer $x > 0$,
\begin{equation*}
\Pr\left( \sum_{k=1}^\infty \mathbb{I}(X_k \le x) \ge \gamma x/\mu \right) \le 2 \exp\left(-\frac{\mu x(\gamma-2)}{2}\right),
\end{equation*}
where $\mathbb{I}$ is the indicator function.
\end{lemma}
\begin{proof}
Let $Y_k = \mathbb{I}(X_k \le x)$.  From \lemref{BiasedWalkPosition},
\begin{equation*}
\Pr(Y_k = 0) \le \exp \left( -\frac{(k \mu - x)^2}{2 k} \right)
\end{equation*}
for $k \le x/\mu$
and
\begin{equation*}
\Pr(Y_k = 1) \le \exp \left( -\frac{(k \mu - x)^2}{2 k} \right)
\end{equation*}
for $k \ge x/\mu$.

Then the quantity to evaluate is
\begin{equation*}
\Pr\left( \sum_{k=1}^\infty Y_k \ge \gamma x / \mu \right).
\end{equation*}
We use a standard trick to split this into two mutually exclusive possibilities and then bound the probabilities separately.  Write
\begin{multline}
\Pr\left( \sum_{k=1}^\infty Y_k \ge \gamma x / \mu \right) = \\
\Pr\left( \left(\sum_{k=1}^\infty Y_k \ge \gamma x / \mu\right) \bigcap \left( \bigcap_{j=1}^{\gamma x/\mu} \left[Y_j=1\right]\right)\right) + \Pr\left( \left(\sum_{k=1}^\infty Y_k \ge \gamma x / \mu\right) \bigcap \left( \bigcup_{j=1}^{\gamma x/\mu} \left[Y_j=0\right]\right)\right).
\end{multline}
We can bound the first term:
\begin{align*}
\Pr\left( \left(\sum_{k=1}^\infty Y_k \ge \gamma x / \mu\right) \bigcap \left( \bigcap_{j=1}^{\gamma x/\mu} \left[Y_j=1\right]\right)\right) &= \Pr\left(\bigcap_{k=1}^{\gamma x/\mu} Y_k=1\right) \\
&\le \Pr\left( Y_{\gamma x/\mu} = 1 \right) \\
&\le \exp\left( -\frac{\mu x(\gamma-1)^2 }{2\gamma}\right)\\
&\le \exp\left( -\frac{\mu x(\gamma -2)}{2} \right)
\end{align*}
The second term similarly:
\begin{align*}
\Pr\left( \left(\sum_{k=1}^\infty Y_k \ge \gamma x / \mu\right) \bigcap \left( \bigcup_{j=1}^{\gamma x/\mu} \left[Y_j=0\right]\right)\right) &\le \Pr\left( \bigcup_{k=\frac{\gamma x}{\mu}+1}^\infty \left[Y_k=1\right] \right) \\
&\le \sum_{k=\frac{\gamma x}{\mu}+1}^\infty \Pr\left(Y_k=1\right) \\
&\le \sum_{k=\frac{\gamma x}{\mu}+1}^\infty \exp\left( -\frac{(k \mu - x)^2}{2k}\right) \\
&\le \exp\left(-\frac{\mu x (\gamma-2)}{2} \right)\qedhere
\end{align*}
\end{proof}
The last fact we need about asymmetric simple random walks is a bound on the probability of going backwards.  If $p>q$ then we expect the walk to go right in the majority of steps.  The probability of going left a distance $a$ is exponentially small in $a$.  This is a well known result, often stated as part of the gambler's ruin problem:
\begin{lemma}[See e.g.~\cite{GrimmettWelsh86}]
\label{lem:NotBackwards}
Consider an asymmetric simple random walk that starts at $a>0$ and has an absorbing barrier at the origin.  The probability that the walk eventually absorbs at the origin is $1$ if $p \le q$ and $\left( q/p \right)^a$ otherwise.
\end{lemma}
This result is for infinitely many steps.  If we only consider finitely many steps, the probability of absorption must be at most this.

\subsubsection{Waiting Time}

From above we saw that the probability of moving is at least $2x/5n$ when at position $x$.  The length of time spent waiting at each step is therefore stochastically dominated by a geometric distribution with parameter $2x/5n$.  The following concentration result will be used to bound the waiting time (in our case $\beta=2/5$):
\begin{lemma}
\label{lem:WaitingConc}
Let the waiting time at each site be $W(x) \sim Geo\left(\beta x/n\right)$, the total waiting time $W = \sum_{x=1}^t W(x)$ and $t' = \frac{n \ln t}{\beta}$.  Then
\begin{equation*}
\Pr(W \ge C t') \le 2 t^{(1-C)/2}.
\end{equation*}
\end{lemma}
\begin{proof}
By Markov's inequality for $\lambda > 0$,
\begin{equation*}
\Pr(W \ge C t') \le \frac{\mathbb{E} e^{\lambda W}}{e^{\lambda C t'}}.
\end{equation*}
The $W(x)$ are independent so
\begin{equation*}
\mathbb{E}e^{\lambda W} = \prod_{x=1}^t \mathbb{E} e^{\lambda W(x)}.
\end{equation*}
Summing the geometric series we find
\begin{equation*}
\mathbb{E} e^{\lambda W(x)} = \frac{\frac{\beta x}{n}}{e^{-\lambda}-1+\frac{\beta x}{n}}
\end{equation*}
provided $e^{\lambda} < \frac{1}{1-\frac{\beta x}{n}}$ for all $1 \le x \le t$.  Therefore $e^\lambda$ is of the form $\frac{1}{1-\frac{\alpha \beta}{n}}$ where $0 < \alpha < 1$.  With this,
\begin{equation*}
\mathbb{E} e^{\lambda W(x)} = \frac{x}{x-\alpha}
\end{equation*}
and
\begin{equation*}
\mathbb{E} e^{\lambda W} = \frac{t! \Gamma(1-\alpha)}{\Gamma(t+1-\alpha)}.
\end{equation*}
We are free to choose $\alpha$ within its range to optimise the bound.  However, for simplicity, we will choose $\alpha = 1/2$.  From \lemref{GammaMGF},
\begin{equation*}
\mathbb{E} e^{\lambda W} \le 2 \sqrt{t}.
\end{equation*}
The result follows, using the inequality $1-x \le e^{-x}$.
\end{proof}

\subsubsection{Phase 1}

Here we prove that phase 1 completes successfully with high probability.  The bias here is large so the walk moves right every time with high probability:
\begin{lemma}
\label{lem:Phase1MovesRight}
The probability that the accelerated chain moves right at each step, starting from $x=1$ for $t$ steps, is at least
\begin{equation*}
1 - t^2/n.
\end{equation*}
\end{lemma}
\begin{proof}
The probability of moving right at each step is
\begin{align*}
\prod_{x=1}^t \frac{3(n-x)}{3n-2x-1} &= \frac{(n-2)(n-3) \ldots (n-t)}{(n-5/3)(n-7/3) \ldots (n-(2t+1)/3)} \\
&\ge (1-2/n)(1-3/n) \ldots (1-t/n) \\
&\ge (1-t/n)^t \ge 1-t^2/n\qedhere
\end{align*}
\end{proof}
Let $t=n^\delta$.  Provided $\delta < 1/2$ this probability is close to one.  Therefore, with high probability, the walk moves to $n^\delta$ in $n^\delta$ steps.  Using \lemref{WaitingConc} the waiting time can be bounded:
\begin{lemma}
\label{lem:Phase1WaitingTime}
Let $W^{(1)}$ be the waiting time during phase 1.  Let $H$ be the event that the walk moves right at each step.  Then
\begin{equation}
\Pr\left(W^{(1)} \ge C t' | H\right) \le 2 n^{\delta(1-C)/2}
\end{equation}
where $t' = \frac{5 \delta n \ln n}{2}$.
\end{lemma}
\begin{proof}
This follows directly from \lemref{WaitingConc}, since each site is hit exactly once.
\end{proof}

We now combine these two lemmas to prove that phase 1 completes successfully with high probability:
\begin{proof}[Proof of \lemref{Phase1Completes}]
In \lemref{Phase1MovesRight}, we show that in $n^\delta$ accelerated steps, the walk moves right at each step with probability $\ge 1 - n^{2\delta-1}$.  Call this event $H$.  Then $\Pr(H) \ge 1 - n^{2\delta - 1}$.  \lemref{Phase1WaitingTime} shows that the waiting time $W^{(1)}$ is bounded with high probability (choosing $C=3$):
\begin{equation*}
\Pr(W^{(1)} \le 15 n \delta \ln n/2 | H) \ge 1-2n^{-\delta}.
\end{equation*}
Then we can bound the probability of phase 1 completing successfully:
\begin{align*}
\Pr(\text{Phase 1 completes successfully}) &\ge \Pr(H \cap W^{(1)} \le 15 n \delta \ln n/2 ) \\
&= \Pr(H) \Pr(W^{(1)} \le 15 n \delta \ln n/2 | H) \\
&\ge (1-n^{2\delta - 1}) (1 - 2n^{-\delta}) \\
&\ge 1 - n^{2\delta-1} - 2 n^{-\delta}.\qedhere
\end{align*}
\end{proof}

\subsubsection{Phase 2}

Phase 2 starts at $n^\delta/2$ and finishes when the walk has reached $\theta n$ for some constant $0<\theta<3/4$.  We show that, with high probability, this also takes time $O(n \log n)$.  The probability of moving right during this phase is at least $p = \frac{3(1-\theta)}{3-2\theta}$.  We first define some constants that we will derive bounds in terms of.  Let $\gamma$ be a constant $>2$.  Let $\mu = p - (1-p)$ and $\tilde{\mu} = \mu/\gamma$.  Finally let $s = \tilde{\mu} t$ for some $t$ (which will be the number of accelerated steps).  Then, with high probability, the walk will have passed $s$ after $t$ steps:
\begin{lemma}
\label{lem:Phase2GetsToRightPlace}
Let $X_t$ be the position of the walk at accelerated step $t$, where $X_0 = n^\delta$.  Then
\begin{equation*}
\Pr(X_t \le s) \le \exp(-\mu^2 t (1-1/\gamma)^2/2).
\end{equation*}
\end{lemma}
\begin{proof}
Let $X_t' = X_t - n^\delta$.  Then from \lemref{BiasedWalkPositionWithBoundedBias},
\begin{equation*}
\Pr(X_t' \le \mu t - \eta) \le \exp\left(-\frac{\eta^2}{2 t}\right).
\end{equation*}
Now let $\eta = \mu t - s$ and use
\begin{align*}
\Pr(X_t \le s) &= \Pr(X_t' \le s - n^\delta) \\
&\le \Pr(X_t' \le s)
\end{align*}
to complete the proof.
\end{proof}
We now prove a bound on the waiting time:
\begin{lemma}
\label{lem:Phase2WaitingTime}
Let $W^{(2)}$ be the waiting time in phase 2.  Then, assuming the walk does not go back beyond $n^\delta/2$,
\begin{equation}
\Pr\left(W^{(2)} \ge \frac{15 n \ln s}{\mu} \right) \le (4/s)^{3/2\mu} + 
\frac{2 \exp \left(\frac{-\mu n^\delta}{4}\right)}{1-\exp \left(\frac{-\mu}{2} \right)}.
\end{equation}
\end{lemma}
\begin{proof}
Let $W_k \sim Geo\left(\frac{2 X_k}{5n}\right)$ where $X_k$ is the position of the walk at accelerated step $k$ ($X_0 = n^\delta$).  We want to bound (w.h.p.) the waiting time $W^{(2)} = \sum_{k=1}^t W_k$ of $t$ steps of the accelerated walk.  

Define the event $H$ to be
\begin{equation}
H = \left\{ \bigcap_{x \ge n^\delta/2} \left[ \sum_{k=1}^{\infty} \mathbb{I}(X_k \le x) \le x/\tilde{\mu} \right] \right\}.
\end{equation}
If $H$ occurs, no sites have been hit too often and the walk has not gone back further than $n^\delta/2$.  It is important that we also use the restriction that $X_k \ge n^\delta/2$ because the waiting time grows the longer the walk moves back.  However, it is very unlikely that the walk will go backwards (even to $n^\delta/2$).

We now define some more notation to bound the waiting time.  Let $\mathbf{X} = (X_1, X_2, \ldots, X_t)$ be a tuple of positions and let $N_x(\mathbf{X})$ be the number of times that $x$ appears in $\mathbf{X}$ and let $\mathbf{N}(\mathbf{X}) = (N_1(\mathbf{X}), N_2(\mathbf{X}), \ldots, N_n(\mathbf{X}))$.  Then we have $\sum_x N_x(\mathbf{X}) = t$.

As we said above, the waiting time at $x=a$ stochastically dominates the waiting time at $x=b$ for $b \ge a$.  In other words,
\begin{equation}
\label{eq:WaitingTimeStocDom}
W_k \stocdoml W_{k'} \text{ if } X_k \le X_{k'}
\end{equation}
where $X \stocdoml Y$ means that $X$ stochastically dominates $Y$.  Now write the waiting time for all steps
\begin{align}
W^{(2)}(\mathbf{X}) &= \sum_{k=1}^t W_k \nonumber \\
&= \sum_x \sum_{h=1}^{N_x(\mathbf{X})} W_h(x)
\label{eq:WaitingTimeX}
\end{align}
where $W_h(x) \sim Geo\left(\frac{2 x}{5n}\right)$.

If event $H$ occurs, we can put some bounds on $N_x$.  We find that, for all $x \ge n^\delta/2$,
\begin{equation}
\sum_{y = n^\delta/2}^x N_y(\mathbf{X}) \le x/\tilde{\mu}
\end{equation}
and $N_x(\mathbf{X}) = 0$ for $x < n^\delta/2$.  Now let $\mathbf{X}_m$ be such that $N_{n^\delta/2}(\mathbf{X}_m) = \frac{n^\delta}{2 \tilde{\mu}}$ and $N_x(\mathbf{X}_m) = 1/\tilde{\mu}$ for $x > n^\delta/2$.  Then
\begin{equation}
\sum_{y = n^\delta/2}^x N_y(\mathbf{X}_m) = x/\tilde{\mu}.
\end{equation}
Now we introduce the relation $\majr$:
\begin{definition}
Let $\mathbf{x}$ and $\mathbf{y}$ be $n$-tuples.  Then $x \majr y$ if
\be
\sum_{i = 1}^k x_i \le \sum_{i=1}^k y_i
\ee
for all $1 \le k \le n$ with equality for $k=n$.
\end{definition}
Note that this is like majorisation, except the elements of the tuples are not sorted.
Using this, we find that $\mathbf{N}(\mathbf{X}) \majr \mathbf{N}(\mathbf{X}_m)$ (Using $\sum_y N_y(\mathbf{X}) = \sum_y N_y(\mathbf{X'}) = t$ for all $\mathbf{X}, \mathbf{X'}$.)

If we combine Equations \ref{eq:WaitingTimeStocDom} and \ref{eq:WaitingTimeX} we find that $W^{(2)}(\mathbf{X}) \stocdoml W^{(2)}(\mathbf{X}')$ if $\mathbf{N}(\mathbf{X}) \majl \mathbf{N}(\mathbf{X'})$.  Roughly speaking, this is simply saying that the waiting time is larger if the earlier sites are hit more often.  But since for all $\mathbf{X}$ that satisfy $H$, $\mathbf{X} \majr \mathbf{X}_m$, we have $W^{(2)}(\mathbf{X}) \stocdomr W^{(2)}(\mathbf{X}_m)$ provided $H$ occurs.  We will simplify further by noting that $\mathbf{X_m} \majr \mathbf{X}_0$ where $N_x(\mathbf{X}_0) = 1/\tilde{\mu}$ for $1 \le x \le \tilde{\mu} t = s$ and zero elsewhere.  Therefore
\begin{equation*}
\Pr\left(W^{(2)}(\mathbf{X}) \ge \frac{5C n \ln s}{2\tilde{\mu}} \bigg| H\right) \le \Pr\left(W^{(2)}(\mathbf{X_0}) \ge \frac{5C n \ln s}{2\tilde{\mu}} \right).
\end{equation*}
We can bound this by applying \lemref{WaitingConc}.  Let $W_h = \sum_{x=1}^{s} W_h(x)$.  From \lemref{WaitingConc},
\begin{equation}
\Pr(W_h \ge C t') \le 2 s^{\frac{1-C}{2}}
\end{equation}
where $t' = \frac{5n \ln s}{2}$.  However, we want a bound on $\Pr\left(\sum_{h=1}^{1/\tilde{\mu}} W_h \ge C t' / \tilde{\mu}\right)$.  The same reasoning as in \lemref{WaitingConc} bounds this as
\begin{equation}
\Pr\left(\sum_{h=1}^{1/\tilde{\mu}} W_h \ge C t' / \tilde{\mu}\right) \le \left(2 s^{\frac{1-C}{2}}\right)^{1/\tilde{\mu}}.
\end{equation}
Therefore 
\begin{equation}
\Pr\left(W^{(2)}(\mathbf{X_0}) \ge \frac{5C n \ln s}{2\tilde{\mu}} \right) \le 2^{1/\tilde{\mu}} s^{\frac{(1-C)/2}{\tilde{\mu}}}.
\end{equation}

To complete the proof, we just need to find $\Pr(H^c)$.  We can bound it using the union bound and \lemref{NumberOfHits}:
\begin{align*}
\Pr(H^c) &= \Pr\left(\bigcup_{x=n^\delta/2}^{n} \left[\sum_{k=1}^\infty \mathbb{I}(X_k \le x) > x/\tilde{\mu}\right] \right) \\
&\le \sum_{x=n^\delta/2}^n \Pr\left(\sum_{k=1}^\infty \mathbb{I}(X_k \le x) \ge x/\tilde{\mu}\right) \\
&\le \sum_{x=n^\delta/2}^n 2 \exp\left(\frac{-\mu x(\gamma -2)}{2}\right) \\
&\le \sum_{x=n^\delta/2}^\infty 2 \exp\left(\frac{-\mu x(\gamma -2)}{2}\right) \\
&= \frac{2 \exp\left(\frac{-\mu n^\delta(\gamma-2)}{4}\right)}{1-\exp\left(\frac{-\mu (\gamma-2)}{2}\right)} \\
\end{align*}
Now, for any events $A$ and $B$
\begin{align*}
\Pr(A) &= \Pr(A \cap B) + \Pr(A \cap B^c) \\
&= \Pr(A | B) \Pr(B) + \Pr(A \cap B^c) \\
&\le \Pr(A | B) + \Pr(B^c)
\end{align*}
and set $C=2$ and $\gamma=3$ to obtain the result.
\end{proof}

We now combine these two lemmas to prove that phase 2 completes successfully with high probability:
\begin{proof}[Proof of \lemref{Phase2Completes}]
Phase 2 can fail if:
\begin{itemize}
\item{The walk does not reach $\theta n$.  The probability of this is bounded by \lemref{Phase2GetsToRightPlace}:
\begin{equation*}
\Pr(X_t \le \theta n) \le \exp\left(-\frac{2}{3} \mu \theta n\right).
\end{equation*}
This follows from setting $t = \frac{3 \theta n}{\mu}$ and $\gamma = 3$.}
\item{The waiting time is too long.  This probability is bounded by \lemref{Phase2WaitingTime}:
\begin{equation*}
\Pr\left(W^{(2)} \ge \frac{15n \ln(\theta n)}{\mu}\right) \le \left(\frac{4}{\theta n}\right)^\frac{3}{2\mu} + \frac{2 \exp\left(\frac{-\mu n^\delta}{4}\right)}{1-\exp(-\mu/2)} + (q/p)^{n^\delta/2}.
\end{equation*}}
\item{The walk gets back to $n^\delta/2$.  This is bounded by \lemref{NotBackwards}:
\begin{equation*}
\Pr\left(\text{Walk gets to $n^\delta/2$}\right) \le \left(q/p \right)^{n^\delta/2}.
\end{equation*}}
\end{itemize}
So, using the union bound we can bound the overall probability of failure:
\begin{equation*}
\Pr(\text{Phase 2 fails}) \le \exp\left(-\frac{2}{3} \mu \theta n\right) + \left(\frac{4}{\theta n}\right)^\frac{3}{2\mu} + \frac{2 \exp\left(\frac{-\mu n^\delta}{4}\right)}{1-\exp(-\mu/2)} + \left(q/p \right)^{n^\delta/2}.\qedhere
\end{equation*}
\end{proof}

\subsubsection{Phase 3}

This phase starts at $\theta n$.  We show that this mixes quickly using log-Sobolev arguments.  

\begin{lemma}
\label{lem:MixingPhase3}
The zero chain on the restricted state space $x \in [m, n]$ where $m = \theta n/2$ for $0 \le \theta \le 3/4$ has mixing time $O\left(n \log \frac{n}{\eps}\right)$.
\end{lemma}
\begin{proof}
We restrict the Markov chain to only run from $m$ by adjusting the holding probability at $m$, $P(m, m)$.  Construct the chain $P'$ with transition matrix
\begin{equation}
P'(x, y)=
\begin{cases}
0 & x < m \, \text{or} \, y < m \\
1 - P(m, m+1) & x = y = m \\
P(x, y) & \text{otherwise}
\end{cases}
\end{equation}
where $P$ is the transition matrix of the full zero chain.  This chain then has stationary distribution
\begin{equation}
\pi'(x)=
\begin{cases}
\pi(x)/(1-w) & m \le x \le n \\
0 & \text{otherwise}
\end{cases}
\end{equation}
where $w = \sum_{x=1}^{m-1} \pi(x)$.  To see this, first note that the distribution is normalised.  We want to show that 
\begin{equation}
\label{eq:Phase3StatDistribDef}
\sum_{x = m}^n P'(x, y) \pi'(x) = \pi'(y).
\end{equation}
When $y = m$ we are required to prove that $P'(m, m) \pi'(m) + P'(m+1, m) \pi'(m+1) = \pi'(m)$.  This follows from the reversibility of the unrestricted zero chain, using $P'(m, m) = 1-P(m, m+1)$.  For $y>m$, \eq{Phase3StatDistribDef} is satisfied simply because $\pi(x)$ is the stationary distribution of $P$ and related by a constant factor to $\pi'(x)$.

We can now prove this final mixing time result, making use of \lemref{ProductChain}.  Let $Q_i$ be the chain that uniformly mixes site $i$.  This converges in one step and has a log-Sobolev constant independent of $n$; call it $\rho_1$.  Let $Q$ be the chain that chooses a site at random and then uniformly mixes that site.  This is the product chain of the $Q_i$ so, by \lemref{ProductChain}, has gap $1/n$ and Sobolev constant $\rho_Q = \rho_1/n$.  We can construct the zero chain for this and find its Sobolev constant.

The Sobolev constant is defined (\defref{LogSobolev}) in terms of a minimisation over functions on the state space.  For the chain $Q$ we can write
\begin{equation*}
\rho_Q = \inf_{\phi} f(\phi).
\end{equation*}
If we restrict the infimum to be over functions $\phi$ with $\phi(x) = \phi(y)$ for $x$ and $y$ containing the same number of zeroes then we obtain the Sobolev constant for the zero-Q chain, $\rho_{Q_0}$, which is chain which counts the number of zeroes in the full chain Q.  Since taking the infimum over less functions cannot give a smaller value,
\begin{equation*}
\rho_{Q_0} \ge \rho_Q \ge \rho_1/n.
\end{equation*}
We can now compare this chain to the zero-$P$ chain.  The stationary distributions are the same.  The transition matrix for the zero-$Q$ chain is
\begin{equation*}
Q_0(x,y) = 
\begin{cases}
\frac{n+2x}{4n} & y = x \\
\frac{x}{4n} & y = x-1 \\
\frac{3(n-x)}{4n} & y = x+1 \\
0	&	{\rm otherwise} \\
\end{cases}
\end{equation*}
Then construct $Q_0'$ by restricting the space to only run from $m$ in exactly the same was as $P'$ is constructed from $P$.  $Q_0'$ has the same stationary distribution as $P'$.  Now we can perform the comparison.  From \eq{ComparisonWalk}:
\begin{align*}
A &= \max_{a \ge m} \frac{Q_0'(a,a+1)}{P'(a,a+1)} \\
&= \max_{a \ge m} \frac{5(n-1)}{8a} \le \frac{5}{8\theta}.
\end{align*}
Therefore $\rho_{P'} \ge \frac{8 \theta \rho_1}{5n}$.  Exactly the same argument applies to show the gap is $\Omega(1/n)$ so the mixing time is (from \eq{SobolevMixingTime}) $O(n \log \frac{n}{\eps})$.
\end{proof}

Now we can prove that phase 3 completes successfully with high probability:
\begin{proof}[of \lemref{Phase3Completes}]
In \lemref{MixingPhase3}, we show that after $O\left(n \log \frac{n}{\eps}\right)$ steps the chain mixes to distance $\eps$.  We just need to show that the walk goes back to $\theta n/2$ with small probability.  This follows from \lemref{NotBackwards}.
\end{proof}

\subsection{Moment Generating Function Calculations}

The following lemma is needed in the moment generating function calculations.
\begin{lemma}
\label{lem:GammaMGF}
For Integer $s > 0$,
\begin{equation}
\label{eq:GammaS}
\frac{\Gamma(s+1)\Gamma(1/2)}{\Gamma(s+1/2)} \le 2 \sqrt{s}
\end{equation}
\end{lemma}
\begin{proof}
From expanding the $\Gamma$ functions, \eq{GammaS} becomes
\begin{align*}
\frac{s! 2^s}{(2s-1)!!} &= \frac{2 \times 4 \times 6 \times \ldots \times 2(s-1) \times 2s}{1 \times 3 \times 5 \times \ldots \times (2s-3) \times (2s-1)} \\
&= \prod_{x=1}^s \frac{2x}{2x-1}
\end{align*}
We then proceed by induction.  $\prod_{x=1}^1 \frac{2x}{2x-1} = 2$ and by the inductive hypothesis
\begin{equation*}
\prod_{x=1}^{s+1} \frac{2x}{2x-1} \le \frac{2(s+1)}{2(s+1)-1} 2\sqrt{s}.
\end{equation*}
It is easy to show that $\frac{2(s+1)}{2(s+1)-1} \le \sqrt{\frac{s+1}{s}}$ and the result follows.
\end{proof}

\subsection{Mixing Times}

We find bounds for the mixing time above that are valid with high probability.  Below we turn these into full mixing time bounds.
\begin{lemma}
\label{lem:DistanceToStat}
If after $O(n \log n)$ steps the state $v$ of a random walk satisfies
\begin{equation*}
|| v - \pi || \le \delta
\end{equation*}
where $\pi$ is the stationary distribution and $\delta$ is $1/poly(n)$ then the number of steps required to be at most a distance $\eps$ from stationarity is
\begin{equation*}
O\left(n \log \frac{n}{\eps}\right).
\end{equation*}
\end{lemma}
\begin{proof}
Let $s$ be the slowest mixing initial state.  Then, after $t = O(n \log n)$ steps we have at worst the state
\begin{equation*}
(1-\delta)\pi+\delta s
\end{equation*}
and if we repeat $k t$ times $\delta$ becomes $\delta^k$.  So to get a distance $\eps$, $k = \left\lceil\frac{\log \eps}{\log \delta}\right\rceil$.

Now we evaluate the mixing time:
\begin{align*}
k t = O(n \log n) \left\lceil\frac{\log \eps}{\log \delta}\right\rceil &= O(n \log n) \left\lceil\frac{\log 1/\eps}{\log 1/\delta}\right\rceil \\
&= O(n \max(\log n, \log 1/\eps)) \\
&= O\left(n \log \frac{n}{\eps}\right)
\qedhere
\end{align*}
\end{proof}

\end{document}